\documentclass[aps,onecolumn,a4paper,superscriptaddress,longbibliography,nofootinbib,notitlepage]{revtex4-1}
\usepackage{graphicx}
\usepackage{amsmath}
\usepackage{amssymb}
\usepackage{amsthm}
\usepackage[utf8]{inputenc}

\usepackage{color}

\newcommand{\abs}[1]{\left\lvert#1\right\rvert} % absolute value: single vertical bars
\newcommand{\norm}[1]{\left\lVert#1\right\rVert} % norm: double vertical bars

\usepackage[a4paper,total={6.5in,9in}]{geometry}
\usepackage[english]{babel}
\usepackage[T1]{fontenc}
\usepackage{mathrsfs}
\usepackage{braket}
\usepackage{physics}
\usepackage{enumerate}
\usepackage{graphicx}
\usepackage{mathtools}
\usepackage[export]{adjustbox}
\usepackage{microtype}
\theoremstyle{plain}

\usepackage[colorlinks = true,
            linkcolor = red,
            urlcolor  = blue,
            citecolor = blue,
            anchorcolor = red]{hyperref}

\newtheorem{theorem}{Theorem}

\newtheorem{lemma}[theorem]{Lemma}
 
\newtheorem{definition}[theorem]{Definition}  
\newtheorem{remark}[theorem]{Remark}  
\newtheorem*{remark*}{Remark}   %  \newtheorem*{remark}{Remark}  * is for numbering
\renewcommand\qedsymbol{$\blacksquare$}
\newenvironment{proof-of}[1][{\hspace{-\blank}}]{{\medskip\noindent\textit{Proof~{#1}.\ }}}{\hfill\qedsymbol}

\renewcommand{\Tr}{{\operatorname{Tr}\,}}

\newcommand{\id}{{\operatorname{id}}}

\newcommand{\1}{\openone}

\newcommand{\cE}{{\mathcal{E}}}
\newcommand{\cT}{{\mathcal{T}}}

\newcommand{\aw}[1]{{#1}}             %toggle comment/uncomment for no red highlights

\begin{document}

\title{Distributed Compression of Correlated \aw{Classical-Quantum} Sources \protect\\ or: The Price of Ignorance}

\author{Zahra Baghali Khanian}
\affiliation{F\'{\i}sica Te\`{o}rica: Informaci\'{o} i Fen\`{o}mens Qu\`{a}ntics, %
Departament de F\'{\i}sica, Universitat Aut\`{o}noma de Barcelona, 08193 Bellaterra (Barcelona), Spain}
\affiliation{ICFO---Institut de Ci\`{e}ncies Fot\`{o}niques, \protect\\%
Barcelona Institute of Science and Technology, 08860 Castelldefels, Spain}
\email{zbkhanian@gmail.com}

\author{Andreas Winter}
\affiliation{F\'{\i}sica Te\`{o}rica: Informaci\'{o} i Fen\`{o}mens Qu\`{a}ntics, %
Departament de F\'{\i}sica, Universitat Aut\`{o}noma de Barcelona, 08193 Bellaterra (Barcelona), Spain}
\affiliation{ICREA---Instituci\'o Catalana de Recerca i Estudis Avan\c{c}ats, %
Pg.~Lluis Companys, 23, 08010 Barcelona, Spain} 
\email{andreas.winter@uab.cat}

\begin{abstract}
We resume the investigation of the problem of independent 
local compression of correlated quantum sources, the classical case 
of which is covered by the celebrated Slepian-Wolf theorem. 
We focus specifically on \aw{classical-quantum (cq)} sources, for which one edge 
of the rate region, corresponding to \aw{the compression of} the classical
part, using the quantum part as side information at the decoder, 
was previously determined by Devetak and Winter [Phys. Rev. A 68, 042301 (2003)].
Whereas the Devetak-Winter protocol attains a rate-sum equal to the von 
Neumann entropy of the joint source, here we show that the full rate 
region is much more complex, due to the partially quantum nature of
the source. In particular, in the opposite case of compressing the
quantum part of the source, using the classical part as side information
at the decoder, typically the rate sum is strictly larger
than the von Neumann entropy of the total source.

We determine the full rate region in the 
\textit{generic} case, showing that, apart from the Devetak-Winter
point, all other points in the achievable region 
have a rate sum strictly larger than the joint entropy. We can interpret 
the difference as the price paid for the quantum encoder being ignorant 
of the classical side information.
In the general case, we give an achievable rate region, via protocols 
that are built on the decoupling principle, \aw{and the principles of quantum} 
state merging and \aw{quantum} state redistribution. 
Our achievable region is matched almost by a single-letter converse,
which however still involves asymptotic errors and an unbounded
auxiliary system.
\end{abstract}

\date{22 November 2018}

\maketitle

\section{Source and compression model}
\label{introduction}
Data compression can be regarded as the foundation of 
information theory in the treatment of Shannon \cite{Shannon1948}, and it remains one 
of the most fruitful problems to be considered, especially when 
additional constraints on the source, the encoders or the decoder 
are imposed. In particular, the Slepian-Wolf problem of two sources  
correlated in a known way, but subject to separate, local compression \cite{Slepian1973}
has proved to provide a unifying principle for much of Shannon
theory, giving rise to natural information theoretic interpretations
of entropy and conditional entropy, and exhibiting deep 
connections with error correction, channel capacities and 
mutual information (cf.~\cite{csiszar_korner_2011}).
The quantum case has been investigated for two decades, starting with the
second author's PhD thesis \cite{Winter1999} \aw{and subsequently in
\cite{Devetak2003}}, up to the systematic study \cite{Ahn2006}, 
and while we still do not have a complete understanding of the rate region,
it has become clear that the problem is of much higher
complexity than the classical case. The quantum Slepian-Wolf
problem, and specifically quantum data compression with side
information at the decoder, has resulted in many fundamental
advances in quantum information theory, including the protocols
of quantum state merging \cite{Horodecki2007,Abeyesinghe2009} and quantum state 
redistribution \cite{Devetak2008_2}, 
which have given operational meaning to the conditional von 
Neumann entropy, the mutual information and the conditional quantum mutual 
information, respectively. 

A variety of resource models and different tasks have been 
considered over the years: The source and its recovery was
either modelled as an ensemble of pure states (following
Schumacher \cite{Schumacher1995}), or as a pure state between the encoders and a
reference system; the communication resource required was
either counted in qubits communicated, in addition either
allowing or disallowing entanglement, or it was counted in
ebits shared between the agents, but with free classical 
communication. While this latter model has lead to the most
complete picture of the general rate region, in the present
paper we will go back to the original idea \cite{Schumacher1995,Winter1999} 
of quantifying the communication, counted in qubits, between the
encoders and the decoder.

\bigskip
\textbf{Notation.}
We use the following conventions throughout the paper.
\aw{Quantum systems are associated with (finite dimensional) Hilbert spaces $A$, $B$, etc.,
whose dimensions are denoted $|A|$, $|B|$, respectively.}
We identify states with their density operators, and 
we use the notation $\phi= \ketbra{\phi}{\phi}$ as the density 
operator of the pure state vector $\ket{\phi}$. 
The von Neumann entropy is defined as $S(\rho) = - \Tr\rho\log\rho$ 
(throughout this paper, $\log$ denotes by default the binary logarithm,
and its inverse function $\exp$, unless otherwise stated, is also to basis $2$).
Conditional entropy and conditional mutual information, $S(A|B)_{\rho}$ and $I(A:B|C)_{\rho}$,
respectively, are defined in the same way as their classical counterparts: 
\begin{align*}
    S(A|B)_{\rho}   &= S(AB)_\rho-S(B)_{\rho}, \text{ and} \\ 
    I(A:B|C)_{\rho} &= S(A|C)_\rho-S(A|BC)_{\rho}
                     = S(AC)_\rho+S(BC)_\rho-S(ABC)_\rho-S(C)_\rho.
\end{align*}
The fidelity between two states $\rho$ and $\sigma$ is defined as 
\[
 F(\rho, \sigma) = \|\sqrt{\rho}\sqrt{\sigma}\|_1 
                 = \Tr \sqrt{\rho^{\frac{1}{2}} \sigma \rho^{\frac{1}{2}}}.
\] 
It relates to the trace distance in the following well-known way \cite{Fuchs1999}:
\[
  1-F(\rho,\sigma) \leq \frac12\|\rho-\sigma\|_1 \leq \sqrt{1-F(\rho,\sigma)^2}.
\]

\bigskip
The source model we shall consider is a hybrid classical-quantum one,
with two agents, Alice and Bob, whose task is is to compress the 
classical and quantum parts of the source, respectively. They then send their
shares to a decoder, \aw{Debbie}, who has to reconstruct the classical
information with high probability and the quantum information with
high (average) fidelity.

In detail, the source is characterised by a classical source, i.e.~a probability
distribution $p(x)$ on a discrete (in fact: finite) alphabet $\mathcal{X}$
which is observed by Alice, and a family of quantum states $\rho_x$
on a quantum system $B$, given by a Hilbert space of finite dimension $|B|$. 
To define the problem of independent local compression (and
decompression) of such a correlated \aw{classical-quantum} source, we 
shall consider purifications $\psi_x^{BR}$ of the $\rho_x$,
i.e.~$\rho_x^B = \Tr_R \psi_x^{RB}$. Thus the source can be described
compactly by the cq-state
\[
  \omega^{XBR} = \sum_{x \in \mathcal{X}} p(x) \ketbra{x}{x}^X \otimes \ketbra{\psi_x}{\psi_x}^{BR}.
\]
We will be interested in the information theoretic limit of
many copies of $\omega$, i.e.
\[
  \omega^{X^n B^n R^n}
    = \left(\omega^{XBR}\right)^{\otimes n}
    = \sum_{x^n \in \mathcal{X}^n} p(x^n) \ketbra{x^n}{x^n}^{X^n} 
                                   \otimes \ketbra{\psi_{x^n}}{\psi_{x^n}}^{B^nR^n},
\]
where we use the notation
\begin{align*}
  x^n              &= x_1 x_2 \ldots x_n, \\
  \ket{x^n}        &= \ket{x_1} \ket{x_2} \cdots \ket{x_n}, \\
  p(x^n)           &= p(x_1) p(x_2)  \cdots p(x_n), \text{ and} \\
  \ket{\psi_{x^n}} &= \ket{\psi_{x_1}} \ket{\psi_{x_2}} \cdots \ket{\psi_{x_n}}.
\end{align*}

Alice and Bob, receiving their respective 
parts of the source, separately encode these using the most general allowed 
quantum operations; the compressed quantum information, living on 
a certain number of qubits, is passed to the decoder who has to
output, again acting with a quantum operation, an element of $\mathcal{X}$
and a state on $B^n$, in such a way as to attain a low error probability
for $x^n$ and a high-fidelity approximation of the conditional quantum
source state, $\psi_{x^n}^{B^nR^n}$.
We consider two models: unassisted and entanglement-assisted, which we
describe formally in the following 
(see Figs.~\ref{fig:una} and \ref{fig:ea}).

\medskip
\textbf{Unassisted model.}
With probability $p(x^n)$, the source provides Alice and Bob respectively 
with states $\ket{x^n}^{X^n}$ and $\ket{\psi_{x^n}}^{B^nR^n}$.
Alice and Bob then perform their respective
encoding operations $\mathcal{E}_X:X^n \longrightarrow C_X$ and 
$\mathcal{E}_B:B^n \longrightarrow C_B$, 
\aw{respectively,} which are quantum operations, i.e.~completely positive and trace preserving (CPTP)
maps. \aw{Of course, as functions they act on the operators (density matrices) over 
the respective input and output Hilbert spaces. But as there is no risk of confusion,
and not to encumber the notation, we will simply write the Hilbert spaces when
denoting a CPTP map. Note that since $X$ is a classical random variable, $\mathcal{E}_X$
is entirely described by a cq-channel.}
We call $R_X=\frac1n \log|C_X|$ and $R_B=\frac1n \log|C_B|$ 
\aw{the} quantum rates of the compression protocol.
Since Alice and Bob are required to act independently, the joint encoding operation 
is $\mathcal{E}_X \otimes \mathcal{E}_B$. 
The systems $C_X$ and $C_B$ are then sent to \aw{Debbie} who performs
a decoding operation \aw{$\mathcal{D}:C_X C_B \longrightarrow \hat{X}^n\hat{B}^n$}.
We define the \aw{extended source state}
\[
  \omega^{X^n {X'}^n B^n R^n} 
      = \left( \omega^{X{X'}BR}\right)^{\otimes n}
      = \sum_{x^n \in \mathcal{X}^n} p(x^n) \ketbra{x^n}{x^n}^{X^n} \otimes \ketbra{x^n}{x^n}^{{X'}^n} 
                                                                    \otimes \ketbra{\psi_{x^n}}{\psi_{x^n}}^{B^nR^n},
\]
and say the encoding-decoding
scheme has average fidelity $1-\epsilon$ if 
\begin{align} 
  \label{F_QCSW_unassisted1}
  \overline{F} := F\left(\omega^{X^n {X'}^n B^n R^n },
      \left(\mathcal{D} \circ (\mathcal{E}_X \otimes \mathcal{E}_B) \otimes \id_{{X'}^n R^n}\right) \omega^{X^n {X'}^n B^n R^n} \right)   
  \geq 1-\epsilon,
\end{align} 
where $\id_{{X'}^n R^n}$ is the identity (ideal) channel acting on ${X'}^n R^n$.
By the above fidelity definition and the linearity of CPTP maps, 
the average fidelity defined in (\ref{F_QCSW_unassisted1}) \aw{can be expressed equivalently as}
\begin{align} 
  \label{F_QCSW_unassisted2}
  \overline{F} = \sum_{x^n \in \mathcal{X}^n } p(x^n) F \left( \ketbra{x^n}{x^n}^{X^n} \otimes \ketbra{\psi_{x^n}}{\psi_{x^n}}^{B^nR^n} ,(\mathcal{D} \circ (\mathcal{E}_X \otimes \mathcal{E}_B) \otimes \id_{R^n}) \ketbra{x^n}{x^n}^{X^n} \otimes \ketbra{\psi_{x^n}}{\psi_{x^n}}^{B^nR^n} \right). \nonumber
\end{align}

We say that $(R_X,R_B)$ is an (asymptotically) achievable rate pair if 
there exist codes $(\mathcal{E}_X,\mathcal{E}_B,\mathcal{D})$ as above 
for every $n$, with fidelity $\overline{F}$ converging to $1$,
and classical and quantum rates converging to $R_X$ and $R_B$, respectively.
\aw{The rate region is the set of all achievable rate pairs, as a subset of $\mathbb{R}_{\geq 0}^2$.}

\begin{figure}[ht]
  \includegraphics[scale=.4]{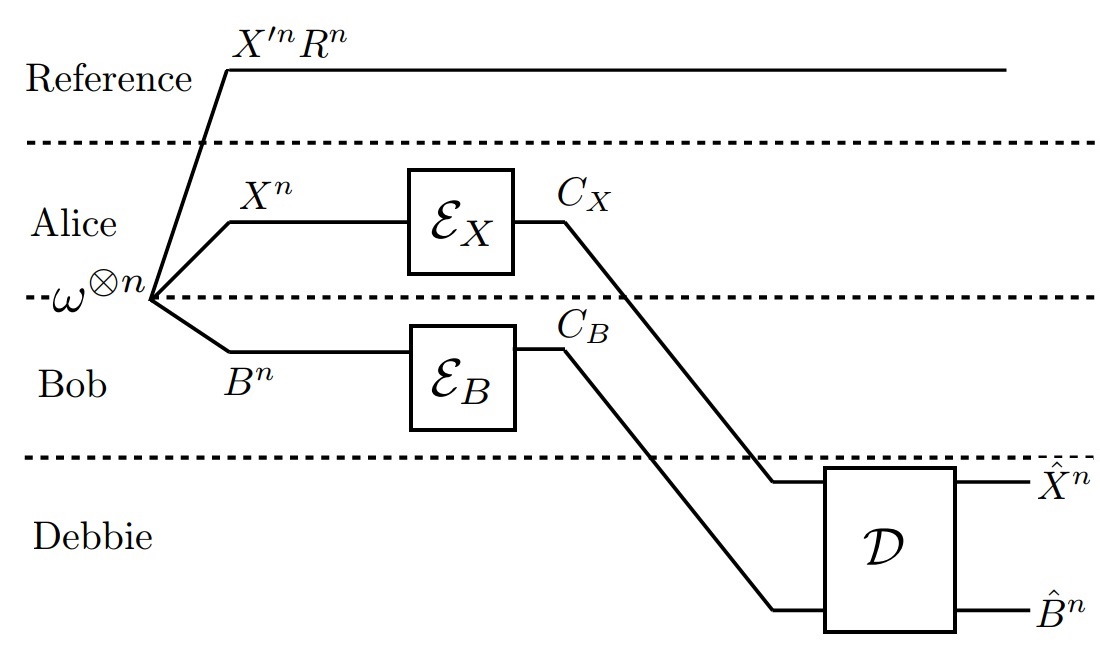}
  \caption{\aw{Circuits diagram of} the unassisted model. Dotted lines are used to 
           demarcate domains controlled by the different participants. 
           The solid lines represent quantum information \aw{registers}.}
  \label{fig:una}
\end{figure}
 
It is shown \aw{by Devetak and Winter} \cite{Devetak2003,Winter1999} that the \aw{rate pair
\begin{equation}
  \label{eq:DW}
  (R_X,R_B) = (S(X|B),S(B))
\end{equation}
is achievable and optimal. The optimality is two-fold; first, the rate sum
achieved, $R_X+R_B=S(XB)$ is minimal, and secondly, even with unlimited $R_B$,
$R_X \geq S(X|B)$. This shows that the Devetak-Winter point is an extreme point
of the rate region. Interestingly,} Alice can achieve the rate $S(X|B)$ using only classical 
communication. However, we \aw{will} prove the converse theorems considering 
a quantum channel for Alice, which are obviously stronger statements.    
In Theorem \ref{theorem: generic full rate region}, we show that our system model is equivalent 
to the model considered in \cite{Devetak2003,Winter1999}, which implies the achievability and 
optimality of this rate pair in our system model. 
\aw{We remark that in \cite{Devetak2003}, the rate $R_B=S(B)$ was not explicitly
discussed, but it is clear that it can always be achieved by Schumacher's quantum
data compression \cite{Schumacher1995}, introducing an arbitrarily small additional error.}

\medskip
\textbf{Entanglement-assisted model.} 
This model \aw{generalizes the unassisted model, and it is basically the same,} 
except that we let Bob and \aw{Debbie} share entanglement \aw{and use it in encoding and decoding, respectively.
In addition, we take care of any possible entanglement that is produced in the process.
Consequently, while Alice's encoding  $\mathcal{E}_X:X^n \longrightarrow C_X$ remains the same,
the Bob's encoding and the decoding map now act as
$\mathcal{E}_B:B^n B_0 \longrightarrow C_B B_0'$ and
$\mathcal{D}:C_X C_B D_0 \longrightarrow \hat{X}^n\hat{B}^n D_0'$, respectively,
where $B_0$ and $D_0$ are $K$-dimensional quantum registers of Bob and \aw{Debbie}, 
respectively, designated to hold the initially shared entangled state, and $B_0'$ and $D_0'$
are $L$-dimensional registers for the entanglement produced by the protocol.
Ideally, both initial and final entanglement are given by maximally
entangled states $\Phi_K$ and $\Phi_L$, respectively.}
Correspondingly, we say \aw{that} the encoding-decoding scheme has average fidelity $1-\epsilon$ if 
\begin{equation} 
  \label{F_QCSW_assisted}
  \overline{F} := F\left(\omega^{X^n {X'}^n B^n R^n }\otimes \Phi_L^{B_0'D_0'},
                    \left(\mathcal{D} \circ (\mathcal{E}_X \otimes \mathcal{E}_{BB_0} \otimes \id_{D_0}) \otimes \id_{{X'}^n R^n}\right) 
                                                                               \omega^{X^n {X'}^n B^n R^n} \otimes \Phi_K^{B_0D_0}\right)
   \geq 1-\epsilon.
\end{equation} 
We call $E=\frac{1}{n}(\log K - \log L)$ the entanglement rate of the scheme.
The CPTP map $\mathcal{E}_{B}$ takes the input systems $B^nB_0$ to the compressed system 
$C_B$ \aw{plus Bob's share of the output entanglement, $B_0'$.}
\aw{Debbie} applies the decoding operation $\mathcal{D}$ on the received systems 
$C_XC_B$ and \aw{her part of the initial} entanglement $D_0$, 
to produce an output state on systems $\hat{X}^n \hat{B}^n$ \aw{plus her share of the output
entanglement, $D_0'$}.
We say $(R_X, R_B, E)$ is an (asymptotically) achievable rate triple if for all $n$
there exist entanglement-assisted codes as before, such that the
fidelity $\overline{F}$ converges to $1$, and
the classical, quantum and entanglement rates converge to
$R_X$, $R_B$ and $E$, respectively.
\aw{The rate region is the set of all achievable rate pairs, as a subset of 
$\mathbb{R}_{\geq 0}^2\times\mathbb{R}$. In the following we will be mostly
interested in the projection of this region onto the first two coordinates,
$R_X$ and $R_B$, corresponding to unlimited entanglement assistance.}

\medskip
\aw{It is a simple consequence of the time sharing principle that the rate regions,
both for the unassisted and the entanglement-assisted model, are closed convex regions.
Furthermore, since one can always waste rate, the rate regions are open to the ``upper right''.
This means that the task of characterizing the rate regions boils down to describing
the lower boundary, which can be achieved by convex inequalities. In the Slepian-Wolf
problem, it is in fact linear inequalities, and we will find analogues of these
in the present investigation.}

\medskip
Stinespring's dilation theorem \aw{\cite{Stinespring1955}} states that any CPTP map can be built 
from the basic operations of isometry and reduction to a subsystem by 
tracing out the environment system \cite{Stinespring1955}. 
Thus, the encoders and the decoder are without loss of generality isometries 
\begin{align*}
    U_X : {X^n} &\longrightarrow {C_X W_X},                          \\
    U_B : {B^n B_0} &\longrightarrow {C_B B_0' W_B},                 \\
    V   : {C_X C_B D_0} &\longrightarrow {\hat{X}^n \hat{B}^n D_0' W_D},
\end{align*}
\aw{where the new systems $W_X$, $W_B$ and $W_D$ are the environment systems
of Alice, Bob and \aw{Debbie}, respectively. They simply remain locally in 
possession of the respective party.}

\begin{figure}[ht] 
  \includegraphics[scale=.4]{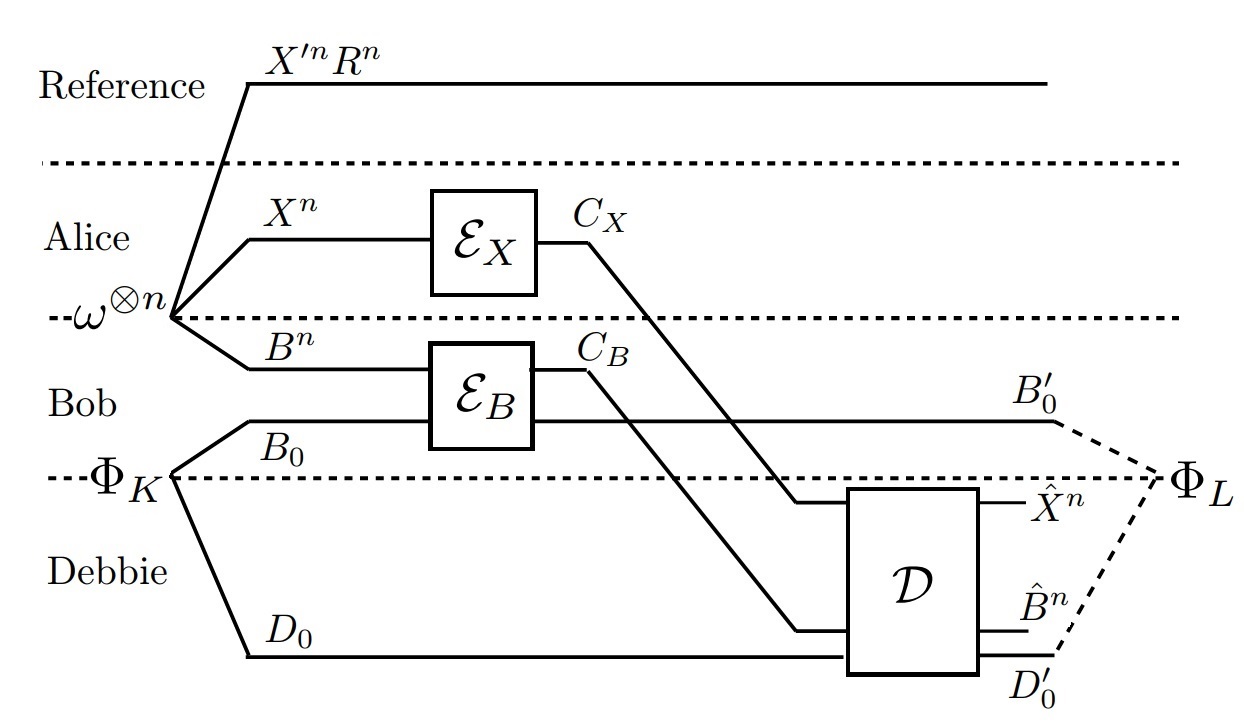} 
  \caption{\aw{Circuits diagram of} the entanglement-assisted model. Dotted lines are used to 
           demarcate domains controlled by the different participants. 
           The solid lines represent quantum information \aw{registers}.}
  \label{fig:ea}
\end{figure}

\aw{The following lemma states that for a code of block length $n$ and error $\epsilon$, 
the environment parts of the encoding and decoding isometries, i.e.~$W_X$, $W_B$
and $W_D$, as well as the entanglement output registers $B_0'$ and $D_0'$, are decoupled from 
the reference $R^n$, conditioned on $X^n$. 
This lemma plays a crucial role in the proofs of converse theorems;
it is proved in Appendix \ref{decoupling_condition_proof}.}

\begin{lemma}(Decoupling condition) 
\label{decoupling condition} 
\aw{For a code of block length $n$ and error $\epsilon$ in the entanglement-assisted model,  
let $W_X$, $W_B$ and $W_D$ be the environments of Alice's and Bob's encoding and of
Debbie's decoding isometries, respectively. Then,
\[
  I(W_XW_BW_D B_0'D_0':\hat{X}^n\hat{B}^nR^n|{X'}^n)_\xi \leq n \delta(n,\epsilon) ,
\] 
where $\delta(n,\epsilon) = 4\sqrt{6\epsilon}\log(|X| |B|) + \frac2n h(\sqrt{6\epsilon})$, 
\aw{with the binary entropy $h(\epsilon)=-\epsilon \log \epsilon -(1-\epsilon)\log (1-\epsilon)$;}
the mutual information is with respect to the state 
\[
  \xi^{{X'}^n \hat{X}^n \hat{B}^n B_0' D_0' W_XW_BW_D R^{n}}
      =\left(\mathcal{D} \circ (\mathcal{E}_X \otimes \mathcal{E}_{B} \otimes \id_{D_0}) \otimes \id_{{X'}^n R^n}\right) 
        \omega^{X^n {X'}^n B^n R^n } \otimes \Phi_K^{B_0D_0}.
\]}
\end{lemma}

\medskip
The \aw{structure of the rest of the paper} is as follows. 
In the next section (Sec.~\ref{sec:seiteninformation}) we start looking
at the important subproblem of compressing the quantum part of the source
when the classical part is sent uncompressed, in other words we
want to find the minimum achievable rate $R_B$ when $R_X$ is
unbounded; this is the opposite edge of the rate region from the
one determined in \cite{Devetak2003,Winter1999}.
We give a general lower (converse) bound and an upper
(achievability) bound, which however do not match in general.
Then, in Sec.~\ref{sec:generic side info} we show that for a family of
\emph{generic} sources, the two bounds coincide, showing that
for almost all sources in any open set of sources, the optimal
quantum compression rate is $R_B = \frac12(S(B)+S(B|X))$.
These results hold in both models, entanglement-assisted and
unassisted. \aw{In Sec.~\ref{sec: full problem}, we move to analysing
the full rate region. We first extend the converse bound from Sec.~\ref{sec:seiteninformation}
to a general outer bound on the rate region (Subsec.~\ref{sec:Converse Bounds in General}),
which yields a tight, single-letter characterization of the rate region for generic sources, 
equally with or without entanglement-assistance (Subsec.~\ref{sec:generic full});
In general, however, can only give an outer bound on the rate region
(Subsec.~\ref{sec:Achievability Bounds in General})
Finally, in Sec.~\ref{sec:discuss}, we close with a discussion of what we
have achieved and of the principal open questions left by our work.}

%%%%%%%%%%%%%%%%%%%%%%%%%%%%%%%%%%%%%%%%%%%%%%%%%%%%%%%%%%%%%%%%%%%%%%%%%%%%%%%%%%%%%%%%%%%%%%

\section{Quantum data compression with classical side information}
\label{sec:seiteninformation}
In this section, we assume that Alice sends her information to \aw{Debbie} 
at rate $R_X=\log \abs{\mathcal{X}}$ such that \aw{Debbie} can decode it 
perfectly, and we ask how much Bob can compress his system given that 
the decoder has access to classical side information $X^n$.  
\aw{This problem is a special case of the \emph{classical-quantum Slepian-Wolf problem}},
and we call it quantum data compression with classical side information at the decoder,
in analogy to the problem of classical data compression 
with quantum side information at the decoder which is addressed 
in \cite{Devetak2003,Winter1999}. Note we do not speak about the
compression and decompression of the classical part at all, and the 
decoder \aw{may} depend directly on $x^n$.
\aw{Of course, by Shannon's data compression theorem \cite{Shannon1948}, $X$ can always be 
compressed to a rate $R_X = H(X)$, introducing an arbitrarily small
error probability}.

We know from previous section that the Bob's encoder is without loss of generality 
an isometry \aw{$U \equiv U_B:{B^n B_0} \longrightarrow {C W B_0'}$, 
taking $B^n$ and Bob's part of the entanglement $B_0$ to systems 
$C\otimes W \otimes B_0'$, where $C \equiv C_B$ is the compressed 
information of rate $R_B=\frac{1}{n}\log |C|$; $W \equiv W_B$ is the environment 
of Bob's encoding CPTP map, and $B_0'$ is the register carrying Bob's share of 
the output entanglement
(in this section, we drop subscript $B$ from $C_B$ and $W_B$). 
Having access to side information $X^n$, \aw{Debbie} applies the decoding isometry  
$V:X^n C D_0 \to \hat{X}^n \hat{B}^n W_D D_0'$ to generate 
the output systems $\hat{X}^n \hat{B}^n$ and entanglement share
$D_0'$, and where $W_D$ is the environment of the isometry.} 
We call this encoding-decoding scheme a side information code of 
block length $n$ and error $\epsilon$ if the average fidelity 
(\ref{F_QCSW_assisted}) is at least $1-\epsilon$.

\subsection{Converse bound}
To state our lower bound on the necessary compression rate, we introduce the
following quantity, which emerges naturally from the converse proof.

\begin{definition}
\label{I_delta}
For the state $\omega^{XBR} = \sum_x p(x) \ketbra{x}{x}^X \otimes \ketbra{\psi_x}{\psi_x}^{BR}$
and $\delta \geq 0$, define
\[
  I_\delta(\omega) := \sup_{\cT} I(X:W)_\sigma 
                  \text{ s.t. } \cT:B\rightarrow W \text{ cptp with } I(R:W|X)_\sigma \leq \delta,
\]
where the mutual informations are understood with respect to the
state $\sigma^{XWR} = (\id_{XR}\otimes \cT)\omega$ \aw{and $W$ ranges over arbitrary
finite dimensional quantum systems}.

The function $I_\delta=I_\delta(\omega)$ is non-decreasing and concave 
in $\delta$. Hence, it is also continuous for $\delta > 0$.
Furthermore, let
\(
  \widetilde{I}_0 := \lim_{\delta \searrow 0} I_\delta = \inf_{\delta>0} I_\delta.
\)
\end{definition}

Note that the system $W$ is not restricted in any way, which is
the reason why in this definition we have a supremum and an infimum, rather 
than a maximum and a minimum. 
(It is a simple consequence of compactness of the domain of optimisation, 
together with the continuity of
the mutual information, that if we were to impose a bound on the dimension of $W$
in the above definition, the supremum in $I_\delta$ would be attained, and
for the infimum in $\widetilde{I}_0$, it would hold that $\widetilde{I}_0 = I_0$.)

\aw{\begin{proof-of}[of the properties of $I_\delta$]
The non-decrease with $\delta$ is evident from the definition, so we only
have to prove concavity. For this consider $\delta_1,\delta_2\geq 0$, $0<p<1$,
and let $\delta = p\delta_1+(1-p)\delta_2$. 
Let furthermore channels $\cT_i:B\rightarrow W_i$ be given ($i=1,2$) such that for the
states $\sigma_i^{XW_iR} = (\id_{XR}\otimes \cT_i)\omega$, $I(R:W_i|X)_{\sigma_i} \leq \delta_i$.
\par
Now define $W := W_1 \oplus W_2$, so that $W_1$ and $W_2$ can be considered
mutually orthogonal subspaces of $W$, and define the new channel
$cT := p\cT_1 + (1-p)\cT_2:B\rightarrow W$. By the chain rule for the mutual
information, one can check that w.r.t.~$\sigma^{XWR} = (\id_{XR}\otimes \cT)\omega$,
\[
  I(R:W|X)_\sigma = p I(R:W_1|X)_{\sigma_1} + (1-p) I(R:W_2|X)_{\sigma_2} \leq p\delta_1+(1-p)\delta_2 = \delta,
\]
and likewise
\[
  I(X:W)_\sigma = p I(X:W_1)_{\sigma_1} + (1-p) I(X:W_2)_{\sigma_2}.
\] 
Hence, $I_\delta \geq p I(X:W_1)_{\sigma_1} + (1-p) I(X:W_2)_{\sigma_2}$; by maximizing over
the channels, the concavity follows.
\end{proof-of}}

\begin{lemma}
  \label{lemma:I-delta}
  The function $I_{\delta}(\omega)$ introduced in Definition \ref{I_delta},
  has the following additivity property. For any two states
  $\omega_1^{X_1B_1R_1}$ and $\omega_2^{X_2B_2R_2}$ and for $\delta,\delta_1,\delta_2 \geq 0$ 
  \[
    I_\delta(\omega_1\otimes\omega_2)
        = \max_{\delta_1+\delta_2= \delta} I_{\delta_1}(\omega_1) + I_{\delta_2}(\omega_2).
  \]
  Consequently, $I_{n\delta}(\omega^{\otimes n}) = n I_\delta(\omega)$, and
  furthermore $I_0$ and $\widetilde{I}_0$ are additive:
  \[
    I_0(\omega_1\otimes\omega_2) = I_0(\omega_1) + I_0(\omega_2),
    \quad
    \widetilde{I}_0(\omega_1\otimes\omega_2) = \widetilde{I}_0(\omega_1) + \widetilde{I}_0(\omega_2).
  \]
\end{lemma}

\begin{proof}
First, we prove that 
$I_\delta(\omega_1\otimes\omega_2) \leq \max_{\delta_1+\delta_2= \delta} I_{\delta_1}(\omega_1) + I_{\delta_2}(\omega_2)$; 
the other direction \aw{of the inequality is trivial from the definition}. 
Let $\cT:B_1 B_2 \to W$ be a CPTP map such that 
\begin{align}
  \label{eq1}
  \delta \geq I(W:R_1R_2|X_1X_2) &=  I(W:R_1|X_1X_2)+I(W:R_2|X_1R_1X_2)  \\ \nonumber
                                 &=  I(WX_2:R_1|X_1)+I(WX_1R_1:R_2|X_2), 
\end{align}
where the second line is due to the independence of $\omega_1$ and $\omega_2$. 
We now define the new systems $W_1:=WX_2$ and $W_2:=WX_1R_1$. \aw{Then we have,}
\begin{align}
  \label{eq2}
  I(W:X_1 X_2) &=    I(W:X_2)+I(W:X_1|X_2)  \\ \nonumber
               &=    I(W:X_2)+I(WX_2:X_1)\\  \nonumber
               &\leq I(\underbrace{WX_1R_1}_{W_2}:X_2)+I(\underbrace{WX_2}_{W_1}:X_1),
\end{align}
where the second equality is due to the independence of $X_1$ and $X_2$. 
The inequality follows \aw{from data processing}.
From \aw{Eq.~(\ref{eq1})} we know that $I(W_1:R_1|X_1)\leq \delta_1$ and $I(W_2:R_2|X_2)\leq \delta_2$ 
for some $\delta_1+\delta_2= \delta$. Thereby, from \aw{Eq.~(\ref{eq2})} we obtain
\begin{align*}
  I_\delta(\omega_1 \otimes \omega_2) &\leq I_{\delta_1}(\omega_1 )+I_{\delta_2}(\omega_2 )\\
                                      &\leq \max_{\delta_1+\delta_2=\delta} I_{\delta_1}(\omega_1 )+I_{\delta_2}(\omega_2 ),
\end{align*}

\aw{Now, the multi-copy additivity follows easily:}
According to the first statement of the \aw{lemma}, we have
\[
  I_{n\delta}(\omega^{\otimes n}) = \max_{\delta_1+\ldots+\delta_n=n\delta} I_{\delta_1}(\omega)+\ldots+I_{\delta_n}(\omega). 
\] 
\aw{Here, the right hand side is clearly $\geq n I_\delta(\omega)$ since we can choose all 
$\delta_i = \delta$. By the concavity of $I_{\delta}(\omega)$ in $\delta$, on the other hand, 
we have for any $\delta_1+\ldots+\delta_n=n\delta$ that
\[
 \frac{1}{n}(I_{\delta_1}(\omega)+\ldots+I_{\delta_n}(\omega)) \leq I_{\delta}(\omega),  
\]
so the maximum is attained at $\delta_i=\delta$ for all $i=1,\ldots,n$}.

The first statement of the \aw{lemma also} implies that $I_0$ and $\widetilde{I}_0$ are additive. 
\end{proof}

\aw{We stop here briefly to remark on the curious resemblance of our function
$I_\delta$ with the so-called \emph{information bottleneck function} introduced 
by Tishby \emph{et al.}~\cite{info-bottleneck}, whose generalization to
quantum information theory is recently being discussed \cite{Salek-QIB,Hirche-QIB}.
Indeed, the concavity and additivity properties of the two functions are proved 
by the same principles, although it is not evident to us, what --if any--, the 
information theoretic link between $I_\delta$ and the information bottleneck is.}

\begin{theorem}
\label{converse_QCSW}
Consider any side information code of block length $n$ and error
$\epsilon$, \aw{in the entanglement-assisted model}. 
Then, the BOb's quantum communication rate is lower bounded
\[
  R_B \geq \frac12 \left( S(B)+S(B|X) - I_{\delta(n,\epsilon)} - \delta(n,\epsilon) \right),
\]
where $\delta(n,\epsilon) = 4\sqrt{6\epsilon}\log(|X| |B|) +\frac2n h(\sqrt{6\epsilon})$.
Any asymptotically achievable rate $R_B$ is consequently lower bounded 
\[
  R_B \geq\frac{1}{2}\left( S(B)+S(B|X)-\widetilde{I}_0 \right).
\]
\end{theorem}

\begin{proof}
%Stinespring's dilation theorem states that any CPTP map can be built  from the basic operations of isometry and reduction to a subsystem by  tracing out the environment system \cite{Stinespring1955}. 
\aw{As already discussed in the introduction to this section,}
the encoder of Bob is without loss of generality 
an isometry $\aw{U:{B^n B_0} \longrightarrow {C W B_0'}}$.
The existence of a high-fidelity 
decoder using $X^n$ as side information %, i.e. an isometry from $C$ to $B^n V$, 
is equivalent to decoupling of $W B_0'$ from 
$R^n$ conditional on \aw{$X^n$; indeed, by Lemma \ref{decoupling condition},} 
$I(R^n:WB_0'|X'^n) \leq n \delta(n,\epsilon)$. 
%where $\delta(n,\epsilon) =2\sqrt{6\epsilon}\log d_R d_{X}^{3}d_{B}^{3} +\frac2n h(\sqrt{6\epsilon})$.
%  
The first part of the converse reasoning is as follows:
\begin{align*}
  nR_B  =    \log |C| 
       &\geq S(C)  \\
       &\geq S(CW B_0')-S(W B_0')  \\
       &=    S(B^n)+S(B_0)-S(W B_0'), 
\end{align*} 
\aw{where the second inequality is a version of subadditivity, and
the equality in the last line holds because the encoding isometry \aw{$U$} 
does not change the entropy; furthermore, $B^n$ and $B_0$ are initially
independent.}
Moreover, the decoder can be dilated to an isometry $\aw{V}: X^n C D_0 \longrightarrow \hat{X}^n \hat{B}^n D_0' W_D$, 
where $W_D$ and $D_0'$ are the environment of \aw{Debbie}'s decoding operation and 
\aw{the output of \aw{Debbie}'s entanglement, respectively.}
Using the decoupling condition of Lemma \ref{decoupling condition} once more, we have
\begin{align*}
 nR_B+S(D_0)&= \log |C| + S(D_0) \\
     &\geq S(C)+S(D_0)  \\
     &\geq S(C D_0)  \\
     &\geq S(X^nCD_0|{X'}^n)  \\
     &= S(\hat{X}^n\hat{B}^n D_0' W_D|{X'}^n) \\
     &= S(W B_0' R^n|{X'}^n) \\
     &\geq S(R^n|{X'}^n)+S(WB_0'|{X'}^n) - n \delta(n,\epsilon) \\
     &= S(B^n|X^n)+S(WB_0'|{X'}^n)       - n \delta(n,\epsilon), 
\end{align*}
\aw{where the third and fourth line are by subadditivity of the entropy;
the fifth line follows because the decoding isometry $V$ does not change the entropy. 
The sixth line holds because for any given $x^n$  
the overall state of the systems $\hat{X}^n\hat{B}^n B_0'D_0' W W_DR^n$ is pure. 
The penultimate line is due to the decoupling condition (Lemma \ref{decoupling condition}), 
and the last line follows because for a given $x^n$ the overall state 
of the systems $B^nR^n$ is pure.}
Adding these two relations and dividing by $2n$, we obtain 
\begin{align*}
  R_B \geq  \frac{1}{2} (S(B)+S(B|X)) - \frac{1}{2n} I(X'^n:WB_0') - \delta(n,\epsilon).
\end{align*}
In the above \aw{inequality, the mutual information on the right hand side} 
is bounded as
\begin{align*}
 I(X'^n:WB_0') \leq I_{n\delta(n,\epsilon)}({\omega^{\otimes n}}) = nI_{\delta(n,\epsilon)}({\omega}), 
\end{align*}
\aw{To see this, define the CPTP map $\mathcal{T}:B^n \longrightarrow \widetilde{W}:= WB_0'$ as
$\mathcal{T}(\rho):= \Tr_{CD_0} (U\otimes\1)(\rho\otimes\Phi_K^{B_0D_0})(U\otimes\1)^\dagger$.
Then we have $I(R^n:\widetilde{W}|{X'}^n) \leq n\delta(n,\epsilon)$, and hence
the above inequality follows directly from Definition \ref{I_delta}.}

The second statement of the theorem follows because
$\delta(n,\epsilon)$ tends to zero as \aw{$n \rightarrow \infty$ and $\epsilon \rightarrow 0$.}
\end{proof}

\begin{remark}
\label{rem:example}
Notice that the term $\frac{1}{n} I({X'}^n:WB_0')$ is not necessarily small. 
For example, suppose \aw{that the source is of the form
$\ket{\psi_x}^{BR} = \ket{\psi_x}^{B'R} \otimes \ket{\psi_x}^{B''}$ for all $x$}; 
clearly it is possible to perform the coding task by
coding only $B'$ and trashing $B''$ (i.e.~putting it into $W$), because 
by having access to $x$ the decoder can reproduce $\psi_x^{B''}$ locally. In 
this setting, characteristically $\frac{1}{n} I({X'}^n:WB_0')$ does not go 
to zero because ${B''}^n$ ends up in $W$.  
\end{remark}

\subsection{Achievable rates}
In this subsection, we provide achievable rates \aw{both for the unassisted and entanglement-assisted} model.

\begin{theorem}
\label{State_merging_rate}
\aw{In the unassisted model, there exists a sequence of side information codes that
compress Bob's system $B^n$} at the asymptotic qubit rate
\begin{align*}
  R_B = \frac{1}{2}\left( S(B)+S(B|X) \right).
\end{align*}
\end{theorem}
\begin{proof}
We can use the fully quantum Slepian-Wolf protocol (FQSW), also called coherent state merging
\cite{Abeyesinghe2009}, as a subprotocol since it considers the entanglement 
fidelity as the decodability criterion, which is more stringent than
the average fidelity defined in (\ref{F_QCSW_unassisted1}). 
Namely, let 
\[
  \ket{\Omega}^{X X^{\prime} B R}
   =\sum_{x \in \mathcal{X}} \sqrt{p(x)} \ket {x}^{X} \ket {x}^{X'} \ket{\psi_{x}}^{B R}
\]
be the source in \aw{the} FQSW problem, \aw{with} the entanglement fidelity $F_e$ is the decodability criterion:
\begin{align*} 
  \label{F_QCSW}
   F_e &= F \left(\Omega^{X^n {X'}^n B^n R^n } ,\left(\mathcal{D} \circ (\id_{X^n} \otimes \mathcal{E}_B) \otimes \id_{{X'}^n  R^n}\right) \Omega^{X^n {X'}^n  B^n R^n } \right)  \nonumber  \\
       &\leq F \left(\omega^{X^n {X'}^n B^n R^n } ,\left(\mathcal{D} \circ (\id_{X^n} \otimes \mathcal{E}_B) \otimes \id_{{X'}^n R^n}\right) \omega^{X^n {X'}^n B^n R^n } \right) 
        = \overline{F},
%&=  \sum_{x^n \in \mathcal{X}^n } p(x^n) F \left(  \ketbra{\psi_{x^n}}{\psi_{x^n}}^{B^nR^n} ,(\mathcal{D} \circ (\mathcal{I}_X \otimes \mathcal{E}_B) \otimes \id_{R^n})  \ketbra{\psi_{x^n}}{\psi_{x^n}}^{B^nR^n} \right)=\overline{F} 
% &=  \sum_{x^n \in \mathcal{X}^n } p(x^n) F \left(  \ketbra{\psi_{x^n}}{\psi_{x^n}}^{B^nR^n} ,\left(  \mathcal{D}_{x^n} \circ \mathcal{E}_B  \otimes \id_{R^n}\right)  \ketbra{\psi_{x^n}}{\psi_{x^n}}^{B^nR^n} \right)=\overline{F} 
\end{align*}
where the inequality is due to the monotonicity of fidelity 
under \aw{CPTP maps, namely the projective measurement on system $X'$ in the computational
basis $\{ \ketbra{x}{x}\}$)}. 
Therefore, if an encoding-decoding scheme attains an entanglement fidelity for 
the FQSW problem going to $1$, then it will have average fidelity for 
the QCSW problem going to $1$ as well. Hence, the FQSW rate 
\begin{align*}
  R_B= \frac{1}{2}I(B:X^{\prime} R)_{\Omega}=\frac{1}{2}(S(B)_{\omega}+S(B|X)_{\omega})
\end{align*}
\aw{is achievable.}
\end{proof}

\begin{remark}
Notice that for the source considered at the end of the previous subsection
\aw{in Remark \ref{rem:example}}, where 
$\ket{\psi_x}^{BR} = \ket{\psi_x}^{B'R} \otimes \ket {\psi_x}^{B''}$ 
for all $x$, we can achieve a rate strictly smaller than the rate 
stated in the above theorem. The reason is that $R$ is only 
entangled with $B'$, so clearly it is possible to perform the coding task by
coding only $B'$ and trashing $B''$ because by having access to $x$ 
the decoder can reproduce the state $\psi_x^{B''}$ locally. 
Thereby, the rate $\frac{1}{2} (S(B')+S(B'|X))$ is achievable by 
applying coherent state merging as above.
\end{remark}

\medskip
The previous observation shows that in general, the rate $\frac12(S(B)+S(B|X))$
from Theorem \ref{State_merging_rate} is not optimal. Looking for a systematic
way of obtaining better rates, we have the following result in the entanglement-assisted
model.

\begin{theorem}
\label{QSR_achievability}
\aw{In the entanglement-assisted model, there exists a sequence of side information codes} 
with the following asymptotic entanglement and qubit rates:
\begin{equation}
  E=\frac{1}{2}\left(I(C:W)_{\sigma}-I(C:X)_{\sigma}\right) 
    \quad \text{and}\quad
  R_B= \frac{1}{2}\left(S(B)_{\omega}+S(B|X)_{\omega}-I(X:W)_{\sigma} \right),   \nonumber
\end{equation}
where $C$ and $W$ are, respectively, the system and environment of an
isometry $V:{B\rightarrow CW}$ on $\omega^{XBR}$ 
producing state $\sigma^{XCWR} = (\id_{XR}\otimes V)\omega$, such that $I(W:R|X)_{\sigma}=0$.
\end{theorem}
\begin{proof}
First, Bob applies \aw{the} isometry $V$ to each copy \aw{of the $n$ systems $B_1,\ldots,B_n$}:  
\begin{align*}
   \sigma^{X X^{\prime} CWR} 
       &= (V^{B \to C W} \otimes \1_{X X^{\prime}R}) 
             \omega^{X X^{\prime}  BR}
          (V^{B \to C W} \otimes \1_{X X^{\prime}R})^\dagger   \\
       &= \sum_x p(x) \ketbra{x}^{X}\otimes \ketbra{x}^{X^{\prime}}\otimes \ketbra{\phi_{x}}^{CWR}. 
\end{align*}
%such that $I(W:R|X)_{\sigma}=0$. 
Now suppose the following source state, where Bob and \aw{Debbie} respectively hold 
the $CW$ and $X$ systems, and Bob wishes to send system $C$ to \aw{Debbie} while keeping 
$W$ for himself:
\[
  \ket{\Sigma}^{X X^{\prime} CW R}
   =\sum_{x \in \mathcal{X}} \sqrt{p(x)} \ket {x}^{X} \ket {x}^{X'} \ket{\phi_{x}}^{CW R}.
\]
For many copies of the above state, \aw{the} parties can apply \aw{the} quantum state redistribution 
(QSR) protocol \cite{Yard2009,Oppenheim2008} for transmitting $C$, having access to system $W$ as 
side information at the encoder and to $X$ as side information at the decoder. 
According to this protocol, Bob needs a rate of 
$R_B = \frac{1}{2}I(C:X^{\prime} R |X)_{\Sigma}
     = \frac{1}{2}(S(B)_{\omega}+S(B|X)_{\omega}-I(X:W)_{\sigma})$ 
qubits of communication.
\aw{The protocol requires a rate of $\frac{1}{2} I(C:W)_{\Sigma}=\frac{1}{2} I(C:W)_{\sigma}$ 
ebits of entanglement shared between the encoder and decoder, and at the end of the protocol
a rate of $\frac{1}{2} I(C:X)_{\Sigma}=\frac{1}{2} I(C:X)_{\sigma}$ ebits of entanglement is 
distilled between the encoder and the decoder.}
This protocol \aw{attains high} fidelity for \aw{the} state $\Sigma^{X^n {X'}^n C^n W^n R^n }$, 
and consequently for \aw{the} state $\sigma^{X^n {X'}^n  C^n W^n R^n }$ due to \aw{the} monotonicity 
of fidelity under \aw{CPTP maps}:
\begin{align} \label{F_QSR}
  1-\epsilon &\leq F\left({\Sigma}^{X^n {X'}^n  C^n W^n R^n }\!\otimes\! \Phi_L^{B_0'D_0'},
                          \left(\mathcal{D} \circ (\id_{X^n D_0} \otimes \mathcal{E}_{CW B_0}) \otimes \id_{{X'}^n R^n}\right) 
                                                         \Sigma^{X^n {X'}^n C^n W^n R^n } \otimes \Phi_K^{B_0D_0}\right) \nonumber \\
             &\leq F\left({\sigma}^{X^n {X'}^n C^n W^n R^n }\!\otimes\! \Phi_L^{B_0'D_0'},
                          \left(\mathcal{D} \circ (\id_{X^n D_0} \otimes \mathcal{E}_{CW B_0}) \otimes \id_{{X'}^n R^n}\right) 
                                                         \sigma^{X^n {X'}^n C^n W^n R^n } \otimes \Phi_K^{B_0D_0}\right),
\end{align} 
where $\mathcal{E}_{CW B_0}$ and $\mathcal{D}$ are respectively the 
encoding and decoding operations of the QSR protocol. 
The condition $I(W:R|X)_{\sigma}=0$ implies that for every $x$ the systems $W$ and $R$ are decoupled:  
\begin{align*}\label{Decoupling_0} 
\phi_{x}^{WR}=\phi_{x}^W  \otimes \phi_{x}^R.   
\end{align*} 
By Uhlmann's theorem \cite{UHLMANN1976,Jozsa1994_2}, 
there exist isometries $V_{x}:{C\rightarrow VB}$ for all $x \in \mathcal{X}$, such that 
\[ 
  (\1\otimes V_{x}^{C\rightarrow VB}) \ket{\phi_{x}}^{CWR}
             =\ket{\nu_{x}}^{VW}  \otimes \ket{\psi_{x}}^{BR}.   
\]
After applying the decoding operation $\mathcal{D}$ of QSR, 
\aw{Debbie} applies the isometry $V_{x}:{C\rightarrow VB}$ for each $x$, 
which does not change the fidelity (\ref{F_QSR}).
By tracing out the unwanted systems $V^n W^n$, 
due to the monotonicity of \aw{the} fidelity under partial trace, 
the fidelity defined in (\ref{F_QCSW_assisted}) \aw{will go to $1$} in this encoding-decoding scheme. 
%Finally, we evaluate the entanglement rate required in the protocol,
%which is $E=\frac12 I(C:W)_{\sigma} - \frac12 I(C:X)_{\sigma}$, and the communication rate,
%which is $R_B = \frac{1}{2} I(C: RX'|X)_{\Sigma}$,
%is equal to $\frac{1}{2}\left(S(B)_{\omega}+S(B|X)_{\omega}-I(X:W)_{\sigma} \right)$, as claimed.
\end{proof}

%At the of the protocol, (approx) the parties are left with n copies of $\sum_x p(x) \ketbra{x}{x} \otimes \ketbra{x}{x} \otimes \psi_x^{BR} \otimes \nu_x^{VW}$

%\begin{definition}\label{I_0}
%For the state $\omega^{XBR}= \sum_x p(x) \ketbra{x}{x}^X \otimes \ketbra{\psi_x}{\psi_x}^{BR}$, define
%\begin{align*}
%    I_0:= \sup I(X:W), 
%\end{align*}
%where the isometry $V^{B \to CW}$ is such that $I(R:W|X) =0$. 
%\end{definition}

\begin{remark}
%(\textcolor{red}{we cannot achieve this rate because QSR works for finite dimensional $W$})
In Theorem \ref{QSR_achievability}, the smallest achievable rate, 
when unlimited entanglement is available,
is equal to $\frac{1}{2}(S(B)+S(B|X)-I_0)$.
This rate resembles the converse bound 
$R_B \geq \frac{1}{2}(S(B)+S(B|X)-\widetilde{I}_0)$,
except \aw{that $\widetilde{I}_0 \geq I_0$}. 
%the fact that for the converse the isometry is such that $I(W:R|X) \leq \epsilon$. 
In the definition of $\widetilde{I}_0$, it seems unlikely that we can take the 
limit  of $\delta$ going to 0 directly because there is no dimension bound on 
the systems $C$ and $W$, so compactness cannot be used directly to prove that 
$\widetilde{I}_0$ and $I_0$ are equal. 
\end{remark}

\aw{\begin{remark}
Looking again at the entanglement rate in Theorem \ref{QSR_achievability},
$E=\frac{1}{2}\left(I(C:W)_{\sigma}-I(C:X)_{\sigma}\right)$, we reflect that
there may easily be situations where $E\leq 0$, meaning that no entanglement is
consumed, and in fact no initial entanglement is necessary. In this case,
the theorem improves the rate of Theorem \ref{State_merging_rate} by the
amount $\frac12 I(X:W)$. 
This motivates the definition of the following variant of $I_0$,
\[
  I_{0-}(\omega) := \sup I(X:W) \text{ s.t. } I(R:W|X)=0,\ I(C:W)-I(C:X) \leq 0,
\]
where the supremum is over all isometries $V:B\rightarrow CW$. 
\par
As a corollary to these considerations, in the unassisted model
the rate $\frac{1}{2}\left(S(B)+S(B|X)-I_{0-} \right)$ is achievable.
\end{remark}}

\subsection{Optimal compression rate for generic sources}
\label{sec:generic side info}
%AW: As i had mentioned, i conjecture that a generic source (i.e. all except a 
%measure zero set of eligible omega's) necessarily has $1/n I(X^n:W) --> 0$, which 
%ought to follow from (*), as it  is the only clear necessary (and in fact - almost - sufficient) 
%condition for a functioning code. {\color{red} In fact, i think that this should be 
%the case whenever all $\psi_x^B$ have the same support, w.l.o.g. the entirety of B. 
%It should be possible to show that this goes to zero for all encodings that 
%asymptotically do not perturb the reduced state of the reference. I somehow feel 
%that a variation of the Koashi-Imoto approach, set out in ...

In this subsection, we find the optimal compression rate for  
\emph{generic} sources, by which we mean any source except for a
submanifold of lower dimension within the set of all sources.
Concretely, we will consider sources where there is at least one $x$ 
for which the reduced state $\psi_x^B= \Tr_R \ketbra{\psi_x}{\psi_x}^{BR}$ 
has full support on $B$. 
%More generally, we can allow any source with support-intersection graph of the $\psi_x^B$ connected.
In this setting, coherent state merging as a subprotocol gives the optimal 
compression rate, so not only does the protocol not use any initial 
entanglement, but some entanglement is distilled at the end of the protocol.

\begin{theorem}
\label{theorem:generic optimal rate}
For any side information code of a generic source, with or without entanglement-assistance,
the asymptotic compression rate $R_B$ of Bob is lower bounded
\begin{align*}
  R_B \geq \frac{1}{2}\left(S(B)+S(B|X)\right),
\end{align*}
so the protocol of Theorem \ref{State_merging_rate} has optimal rate for a generic source.
Moreover, in that protocol no prior entanglement is needed and 
a rate $\frac{1}{2}I(X:B)$ ebits of entanglement 
is distilled between the encoder and decoder.
\end{theorem}

\begin{proof}
The converse bound of Theorem \ref{converse_QCSW} states that 
\aw{the} asymptotic quantum communication rate of Bob is lower bounded as
\begin{align}
  R_B \geq \frac{1}{2}\left(S(B)+S(B|X)-\widetilde{I}_0 \right),   \nonumber
\end{align}
where $\widetilde{I}_0$ \aw{comes from} Definition \ref{I_delta}. W
e will show that for generic sources, $\widetilde{I}_0 = I_{0}= 0$.
Moreover, Theorem \ref{State_merging_rate} states that using coherent state 
merging, the asymptotic qubit rate of $\frac{1}{2}(S(B)+S(B|X))$ is achievable, 
that no prior entanglement is required and a rate of 
$\frac{1}{2}I(X:B)$ ebits of entanglement is distilled between the encoder and the decoder.

%To do so, we have to study $I_\delta$ for asymptotically vanishing $\delta > 0$.
We show that for any CPTP map $\cT:B\to W$, which acts on a generic $\omega^{XBR}$ and produces  
state $\sigma^{XWR} = (\id_{XR}\otimes \cT)\omega^{XBR}$ such that $I(R:W|X)_{\sigma}\leq \delta$ 
for  $\delta \geq 0$, the quantum mutual information 
$I(X:W)_{\sigma} \leq \delta' \log |X| +2h (\frac12\delta')$ where $\delta'$ is defined in 
\aw{Eq.~(\ref{phix_phi0_distance}) below}. 
Thus, we obtain
\[ 
  \widetilde{I}_0 = \lim_{\delta \searrow 0} I_\delta = 0.
\]
%By the definition, for all $x$ the system $W$ is almost 
%decoupled from the reference $R$, i.e.
\aw{To show this claim, we proceed as follows.} From $I(R:W|X)_{\sigma}\leq \delta$ we have
\begin{align*}
  I(R:W|X=x)_\sigma \leq \frac{\delta}{p(x)},
\end{align*}
\aw{so} by Pinsker's inequality \cite{Schumacher2002} we obtain
\begin{align*}
  \left\| \phi_x^{W R}- \phi_x^{W} \otimes \phi_x^{R} \right\|_1  \leq \sqrt{\frac{2 \delta \ln 2}{p(x)}}.
\end{align*}
By Uhlmann's theorem, there exists an isometry $V_x:{C \to BV}$ such that
\begin{align}\label{eq3}
%\label{phi_x_distance}
 \left\| (V_x \otimes \1_{WR}) \phi_x^{CW R} (V_x \otimes \1_{WR})^{\dagger}
                                        - \theta_x^{WV} \otimes \psi_x^{BR} \right\|_1  
 \leq \sqrt{\sqrt{ \frac{\delta\ln 2}{2p(x)}} \left(2-\sqrt{ \frac{\delta\ln 2}{2p(x)}}\right)}, 
\end{align}
where $\theta_x^{WV}$ is a purification of $\phi_x^{W}$.
Since the source is generic by definition there is an $x$, say $x=0$, 
for which $\psi_0^B$ has full support on
$\mathcal{L}(\mathcal{H}_B)$, i.e.~$\lambda_0:=\lambda_{\min}(\psi_0^B)>0$. 
%We call this source a \textit{generic} source.  
By Lemma \ref{full_support_lemma} in Appendix \ref{Miscellaneous_Facts}, 
for any $\ket{\psi_x}^{BR}$ there is an operator $T_x$ acting on the reference system such that
\begin{align*}
%  \label{T_x}
  \ket{\psi_x}^{BR} = (\1_B \otimes T_x) \ket{\psi_0}^{BR}. \nonumber
\end{align*}
Using this fact, we show that the decoding isometry $V_0$ in \aw{Eq.~(\ref{eq3})} works for all states:
\begin{align*}
%  \label{phi_0_distance}
  \bigl\| (&V_0 \otimes \1_{WR}) \phi_x^{CWR} (V_0^{\dagger} \otimes \1_{W R})
                            - \theta_0^{WV}  \otimes  \psi_x^{BR} \bigr\|_1   \\
   &= \norm{(V_0 \otimes \1_{WR}) (\1_{CW} \otimes T_x) \phi_0^{CWR} (\1_{CW} \otimes T_x)^{\dagger} (V_0^{\dagger}\otimes \1_{W R})- \theta_0^{WV} \otimes (\1_B \otimes T_x)\psi_0^{BR}(\1_B \otimes T_x)^{\dagger}}_1   \\
   &= \left\| (\1_{BVW}\otimes T_x)(V_0\otimes\1_{WR}) \phi_0^{CWR} (V_0^{\dagger}\otimes\1_{W R}) (\1_{BVW}\otimes T_x^{\dagger}) 
            - (\1_{BVW} \otimes T_x)\theta_0^{WV} \otimes\psi_0^{BR}(\1_{BVW} \otimes T_x^{\dagger}) \right\|_1  \\
   &\leq  \norm{\1_{BVW} \otimes T_x}_{\infty}^2  
          \norm{ (V_0 \otimes \1_{WR}) \phi_0^{CWR}  (V_0^{\dagger}  \otimes \1_{WR}) -  \theta_0^{WV} \otimes\psi_0^{BR}}_1   \\
   &\leq   \frac{1}{\lambda_0}  \sqrt{\sqrt{ \frac{\delta\ln 2}{2p(0)}} \left(2-\sqrt{ \frac{\delta\ln 2}{2p(0)}}\right)}, 
\end{align*}
where the \aw{last two} inequalities follow from Lemma \ref{T_norm1_inequality} 
and Lemma \ref{full_support_lemma}, respectively. 
By tracing out the systems $VBR$ in the above chain of inequalities, we get 
\begin{align}
  \label{phix_phi0_distance}
  \norm{\phi_x^{W}- \phi_0^{W}}_1 
      \leq \frac{1}{\lambda_0}\sqrt{\sqrt{ \frac{\delta\ln 2}{2p(0)}} \left(2-\sqrt{ \frac{\delta\ln 2}{2p(0)}}\right)} \aw{=: \delta'}.   
\end{align}
Thus, by triangle inequality we obtain
\begin{align}
  \label{phi_phi0_distance} 
  \norm{\underbrace{\sum_x p(x) \ketbra{x}{x}^X \otimes \phi_x^{W}}_{\sigma^{XW}}
         - \underbrace{\sum_x p(x) \ketbra{x}{x}^X \otimes \phi_0^{W}}_{\aw{=:}\sigma_0^{XW}}}_1 
    &\leq \sum_x p(x) \norm{  \phi_x^{W}- \phi_0^{W}}_1 \nonumber \\ 
    &\leq \frac{1}{\lambda_0}\sqrt{\sqrt{ \frac{\delta\ln 2}{2p(0)}} 
          \left(2-\sqrt{ \frac{\delta\ln 2}{2p(0)}}\right)} = \delta'.    
\end{align}
By applying \aw{the Alicki-Fannes inequality in the form of Lemma \ref{AFW lemma},} to 
\aw{Eq.}~(\ref{phi_phi0_distance}), we have
\begin{align*}
  I(X:W)_\sigma &=S(X)_{\sigma}-S(X|W)_{\sigma} +S(X|W)_{\sigma_0}-S(X|W)_{\sigma_0} \\
                &= S(X|W)_{\sigma_0}-S(X|W)_{\sigma}                                 \\
                &\leq \delta' \log |X| + 2h\left(\frac12\delta'\right),  
\end{align*}
and the right \aw{hand side} of the above inequality vanishes for \aw{$\delta\rightarrow 0$}.
\end{proof}

%%%%%%%%%%%%%%%%%%%%%%%%%%%%%%%%%%%%%%%%%%%%%%%%%%%%%%%%%%%%%%%%%%%%%%%%%%%%%%%%%%%%%%%%%%%%%%
\aw{\section{Towards the full rate region}}
\label{sec: full problem}
In this section, we consider the full rate region of the distributed compression of 
a \aw{classical-quantum} source. \aw{The Devetak-Winter code, Eq.~(\ref{eq:DW}), and 
the code based on state merging, Theorem \ref{State_merging_rate}, we get two rate
points in the unassisted (and hence also in the unlimited entanglement-assisted)
rate region:
\[
  (R_X,R_B) = (S(X|B),S(B)), \quad
  (R_X,R_B) = \left(S(X),\frac12(S(B)+S(B|X))\right).
\]
Their upper-right convex closure is hence an inner bound to the rate region,
depicted schematically in Fig.~\ref{fig:inner} and described by the inequalities
in the following theorem.

\begin{theorem}
\label{unknown.theorem}
For distributed compression of a \aw{classical-quantum} source in unassisted model, 
the rate pairs satisfying the following inequalities are achievable:
\begin{equation}\begin{split}
  \label{eq:inner}
  R_X      &\geq S(X|B),                             \\
  R_B      &\geq\frac{1}{2}\left(S(B)+S(B|X)\right), \\
  R_X+2R_B &\geq S(B)+S(XB). 
\end{split}\end{equation}
\end{theorem}}

\begin{figure}[ht]
%  \framebox[0.8\textwidth]{\rule{0pt}{160pt}}
  \includegraphics[width=0.9\textwidth]{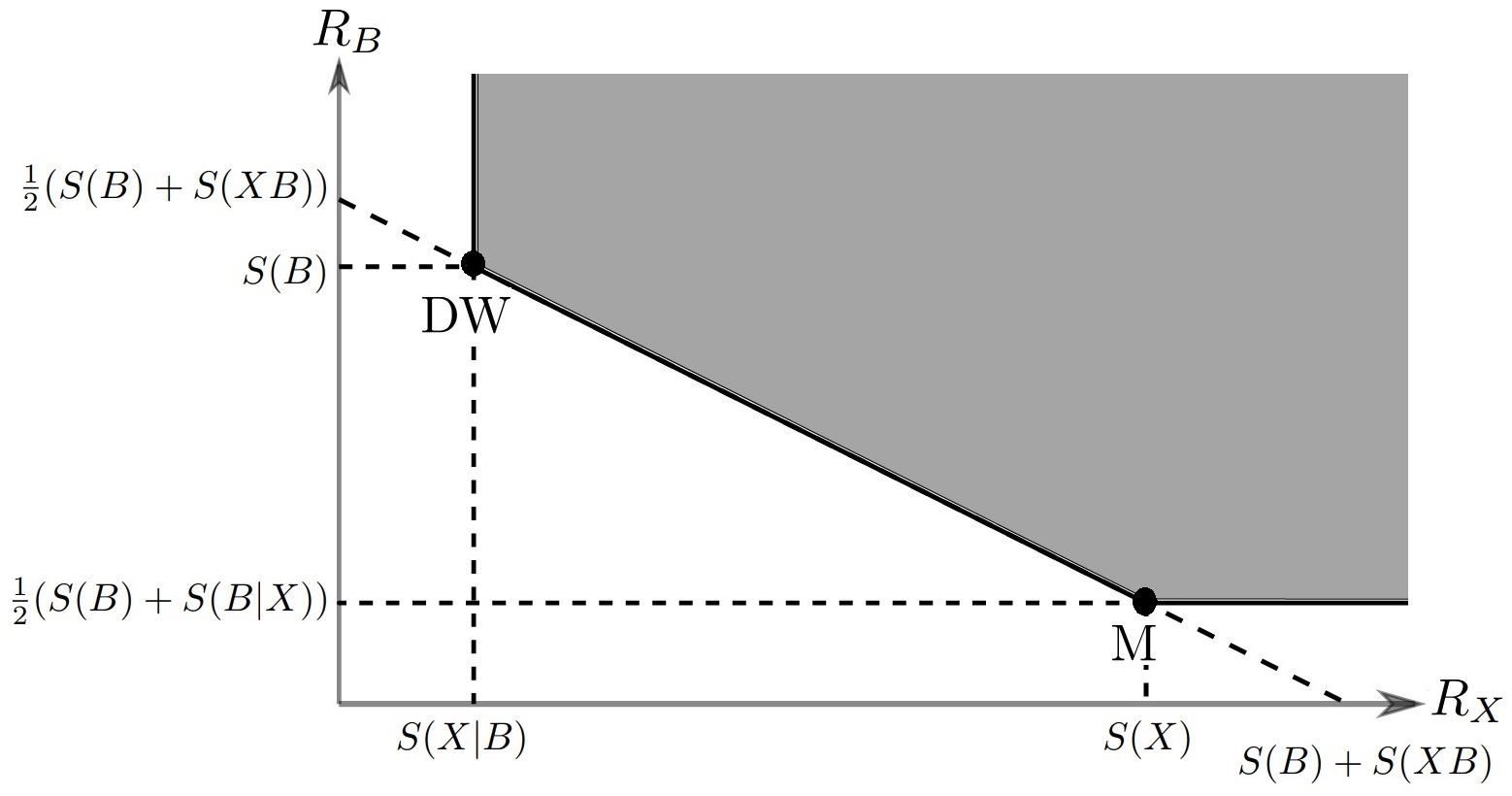}
  \caption{\aw{The region of all pairs $(R_X,R_B)$ satisfying the three conditions of 
           Eq.~(\ref{eq:inner}); it is the upper-right convex closure of the
           Devetak-Winter (DW) and the merging (M) point.
           All of these points are achievable in the unassisted model}.}
  \label{fig:inner}
\end{figure}

\aw{For generic sources we find \aw{that this is in fact the} rate region.
However, in general, we only present some outer bounds and inner bounds (achievable rates), 
which show the rate region to be much more complicated than the rate region of \aw{the}
classical Slepian-Wolf problem.}

\aw{\subsection{General converse bounds}
\label{sec:Converse Bounds in General}
For distributed compression of a \aw{classical-quantum} source 
in general, we start with a general converse bound.

\begin{theorem}
\label{theorem:full converse}
The asymptotic rate pairs for distributed compression of a 
\aw{classical-quantum} source in the entanglement-assisted model are lower bounded as 
\begin{equation}\begin{split}
  \label{eq:general-converse}
  R_X        &\geq S(X|B),                                                \\
  R_B        &\geq \frac{1}{2}\left( S(B)+S(B|X)-\widetilde{I}_0 \right), \\ 
%  R_X +R_B   &\geq  S(XB)-E,                                              \\
  R_X + 2R_B &\geq  S(B)+S(BX)-\widetilde{I}_0.
\end{split}\end{equation}
%where $E=\frac1n (\log K -\log L)$ is the entanglement rate of the protocol. In the unassisted model, all the above lower bounds hold except the third lower bounded which is simplified as
In the unassisted model, in addition to the above lower bounds, 
the asymptotic rate pairs are bounded as
\[ 
  R_X + R_B \geq S(XB). 
\]
\end{theorem}

\begin{proof}
The individual lower bounds have been established already: 
$R_X\geq S(X|B)$ is from \cite{Devetak2003,Winter1999}, in a slightly different source model. 
However, it also holds in our system model if Bob sends his information using 
unlimited communication such that \aw{Debbie} can decode it perfectly. Namely, notice 
that the fidelity (\ref{F_QCSW_unassisted1}) is more stringent than the decoding 
criterion of \cite{Devetak2003,Winter1999}, so any converse bound considering the 
decoding criterion of \cite{Devetak2003,Winter1999} is also a converse bound in our system model. 
%In order to bound the asymptotic communication rate of Alice $R_X$, assume that Bob sends 
%his information to \aw{\aw{Debbie}} at rate $\log \abs{B}$, that is \aw{Debbie} can decode 
%Bob's information perfectly. Thereby, we have 
%\begin{align}\label{eq:lower bound R_X}
%    nR_X &=\log \abs{C_X}\nonumber\\
%    &\geq S(C_X)\nonumber\\
%    &\geq S(C_X|B^n)\nonumber\\
%    &= S(C_XB^n)-S(B^n)\nonumber\\
%    &= S(\hat{X}^n\hat{B}^nW_D)-S(B^n)\nonumber\\
%    &\geq S(\hat{X}^n\hat{B}^n)+S(W_D |\hat{X}^n\hat{B}^nX'^n )-S(B^n)\nonumber\\
%    &\geq S(\hat{X}^n\hat{B}^n)+S(W_D |X'^n )-S(B^n)-O(\sqrt{\epsilon})\nonumber\\
%    &\geq S(\hat{X}^n\hat{B}^n)-S(B^n)-O(\sqrt{\epsilon})\nonumber\nonumber\\
%    &\geq S(X^n B^n)-S(B^n)-O(\sqrt{\epsilon})\nonumber\nonumber\\
%    &= nS(X|B)-O(\sqrt{\epsilon}),
%\end{align}
%where the third equality follows because decoding isometry does not change the entropy. 
%The forth inequality follows from more general form of the decoupling condition 
%(\aw{Lemma} \ref{decoupling_I}) which basically states that for a given $x^n$ systems 
%$\hat{X}^n\hat{B}^n$ and $W_D D_0'$ are \textit{almost} independent. 
%The penultimate inequality follows because conditional entropy is always positive if 
%one of the systems is classical. 
%The last inequality follows from decodability of systems $X^nB^n$, that is, fidelity 
%criterion (\ref{F_QCSW_assisted}) implies that the output state on systems $\hat{X}^n\hat{B}^n$ 
%is very close to the original state $X^nB^n$; the inequality follows by applying the Fannes 
%inequality (Lemma \ref{Fannes-Audenaert inequality}).
The bound $R_B\geq \frac{1}{2}(S(B)+S(B|X)-\widetilde{I}_0)$ is from
Theorem \ref{theorem:generic optimal rate}. These two bounds hold in the unassisted,
as well as the entanglement-assisted model.

In the unassisted model, the rate sum lower bound $R_X + R_B \geq S(XB)$ has been
argued in \cite{Devetak2003,Winter1999}, too. As a matter of fact, for any distributed compression
scheme for the source, $\cE_X\otimes\cE_B$ jointly describes a Schumacher compression scheme 
with asymptotically high fidelity- Thus, its rate must be asymptotically lower bounded
by the joint entropy of the source, $S(XB)$ \cite{Schumacher1995,Jozsa1994_1,Barnum1996,Winter1999}.

This leaves the bound $R_X + 2R_B \geq S(B)+S(BX)-\widetilde{I}_0$ to be proved in
the entanglement-assisted model, which we tackle now.
The encoders of Alice and Bob are isometries $U_X:{X^n \to C_X W_X}$ and 
$U_B:{B^nB_0 \to C_B W_B B_0'}$, respectively. They send their respective compressed systems 
$C_X$ and $C_W$ to \aw{Debbie} and keep the environment parts $W_X$ and $W_B$ for themselves. 
Then, \aw{Debbie} applies the decoding isometry $V:{C_X C_B D_0 \to \hat{X}^n\hat{B}^n W_D D_0'}$,
where systems $\hat{X}^n\hat{B}^n D_0'$ are the output states, and $W_D$ and $D_0'$ are the 
environment of \aw{Debbie}'s decoding isometry and her output entanglement, respectively.
We first bound the following sum rate:
\begin{align}
  \label{eq sum rate}
  nR_X+nR_B+S(D_0) &\geq S(C_X)+S(C_B)+S(D_0) \nonumber\\
    &\geq S(C_XC_BD_0)                        \nonumber\\
    &=    S(\hat{X}^n\hat{B}^n W_D D_0')      \nonumber\\
    &=    S(\hat{X}^n\hat{B}^n) + S(W_D D_0'|\hat{X}^n\hat{B}^n)     \nonumber\\
    &\geq S(\hat{X}^n\hat{B}^n) + S(W_D D_0'|\hat{X}^n\hat{B}^nX'^n) \nonumber\\
    &\geq S(X^n B^n) + S(W_D D_0'|\hat{X}^n\hat{B}^nX'^n) - n\sqrt{2\epsilon} \log(|X| |B|) - h(\sqrt{2\epsilon}) \nonumber\\
    &\geq S(X^n B^n) + S(W_D D_0'|X'^n) - 2n\delta(n,\epsilon)                                    \nonumber\\
    &\geq S(X^n B^n) + S(W_XW_B B_0'|X'^n) - S(R^n\hat{B}^n\hat{X}^n|X'^n) - 2n\delta(n,\epsilon) \nonumber\\
    &\geq S(X^n B^n) + S(W_XW_B B_0'|X'^n) - 3n\delta(n,\epsilon)                                 \nonumber\\
    &=    S(X^n B^n) + S(W_X|X'^n) + S(W_B B_0'|X'^n) - 3n\delta(n,\epsilon)                      \nonumber\\
    &\geq S(X^n B^n) + S(W_B B_0'|X'^n) - 3n\delta(n,\epsilon),
\end{align}
where the second line is by subadditivity, the equality in the third line follows because 
the decoding isometry $V$ does not change the entropy. Then, in the fourth and fifth line
we use the chain rule and strong subadditivity of entropy.
The inequality in the sixth line follows from the decodability of the systems $X^nB^n$:
the fidelity criterion (\ref{F_QCSW_assisted}) implies that the output state on systems 
$\hat{X}^n\hat{B}^n$ is $2\sqrt{2\epsilon}$-close to the original state $X^nB^n$ in trace norm;
then apply \aw{the} Fannes inequality (\aw{Lemma} \ref{Fannes-Audenaert inequality}). 
The seventh line follows from the decoupling condition (Lemma \ref{decoupling condition}),
which implies that 
$I(W_D D_0':\hat{X}^n\hat{B}^n|{X'}^n) \leq n\delta(n,\epsilon) 
                                       =    4n\sqrt{6\epsilon} \log(|X| |B|) + 2 h(\sqrt{6\epsilon})$.
In the eighth line, we use that for any given $x^n$, the overall state of
$W_XW_B W_D B_0' D_0'R^n\hat{B}^n\hat{X}^n$ is pure, and invoking subadditivity;
then, in line nine we use the decoding fidelity (\ref{F_QCSW_assisted}) once
more, saying that the output state on systems 
$\hat{X}^n\hat{B}^nR^n{X'}^n$ is $2\sqrt{2\epsilon}$-close to the original 
state $X^nB^nR^n{X'}^n$ in trace norm;
then apply \aw{the} Fannes inequality (\aw{Lemma} \ref{Fannes-Audenaert inequality}).
The equality in the eleventh line follows because for 
a given $x^n$ the encoded states of Alice and Bob are independent.  

Moreover, we bound $R_B$ as follows:
\begin{align}
  \label{eq RB lower bound}
  nR_B &\geq S(C_B)                       \nonumber\\
       &\geq S(C_B|W_BB_0')               \nonumber\\
       &=    S(C_BW_BB_0') - S(W_BB_0')   \nonumber\\
       &=    S(B^nB_0) - S(W_BB_0')       \nonumber\\
       &=    S(B^n) + S(B_0) - S(W_BB_0'). 
\end{align}
Adding \aw{Eqs.}~(\ref{eq sum rate}) and (\ref{eq RB lower bound}), and after
cancellation of $S(B_0)=S(D_0)$ we get
\begin{align}
  \label{eq: lower bound R_X+2R_B}
  R_X+2R_B &\geq S(B) + S(X B) - \frac{1}{n}I(X'^n:W_B B_0') - 3n\delta(n,\epsilon)                              \nonumber\\
           &\geq S(B) + S(X B) - \frac{1}{n}I_{n\delta(n,\epsilon)}({\omega^{\otimes n}}) - 3n\delta(n,\epsilon) \nonumber\\
           &=    S(B) + S(X B) - I_{\delta(n,\epsilon)}({\omega}) - 3n\delta(n,\epsilon),
\end{align}
where given that $I(R^n:B_0'W_B|X'^n) \leq \delta(n,\epsilon)$, which we have from the 
decoupling condition (Lemma \ref{decoupling condition}), the second equality follows directly 
from Definition \ref{I_delta}, just as in the proof of Theorem \ref{converse_QCSW}. 
The equality in the last line follows from Lemma \ref{lemma:I-delta}. 
In the limit of $n\rightarrow\infty$ and $\epsilon\rightarrow 0$, we have
$\delta(n,\epsilon) \rightarrow 0$, and so $I_{\delta(n,\epsilon)}$ \aw{converges} 
to $\widetilde{I}_0$.
\end{proof}}

\subsection{Rate region for generic sources}
\label{sec:generic full}
In this subsection, we find the complete rate region for generic sources, generalizing
the insight of Theorem \ref{theorem:generic optimal rate} for the subproblem
of quantum compression with classical side information at the decoder.

\begin{theorem}
\label{theorem: generic full rate region}
\aw{For a generic classical-quantum source, in particular one where there is
an $x$ such that $\psi_x^B$ has full support}, the optimal asymptotic rate region for distributed 
compression is \aw{the} set of rate pairs satisfying 
\begin{align*}
  R_X      &\geq S(X|B), \\
  R_B      &\geq\frac{1}{2}\left(S(B)+S(B|X)\right),\\
  R_X+2R_B &\geq S(B)+S(XB). 
\end{align*}
Moreover, \aw{there are protocols achieving these bounds requiring no prior} entanglement. 
%and a rate $\frac{1}{2}I(X:B)$ ebits of entanglement 
%is distilled between the encoder and decoder at the rate point 
%$(R_X,R_B)=(S(X),\frac{1}{2}\left(S(B)+S(B|X)\right))$.
%Moreover, $XXXX$ number of ebits of entanglement is distilled between the encoder and decoder.
\end{theorem}

\aw{\begin{proof}
We have argued the achievability already at the start of this section
(Theorem \ref{unknown.theorem}).
As for the converse, we have shown in Theorem \ref{theorem:generic optimal rate}
that for a generic source, $\widetilde{I}_0=0$, hence the claim follows from
the outer bounds of Theorem \ref{theorem:full converse}.
\end{proof}}

This means that for generic sources, which we recall are the complement of a set of measure
zero, the rate region has the shape of Fig.~\ref{fig:inner}.

\subsection{General achievability bounds}
\label{sec:Achievability Bounds in General}
For general, non-generic sources, the achievability bounds of Theorem \ref{unknown.theorem}
and the outer bounds of Theorem \ref{theorem:full converse} do not match. Here we present several more
general achievability results that go somewhat towards filling in the unknown
area in between, without, however, resolving the question completely.

\begin{theorem}
\label{thm:achieve}
For distributed compression of a \aw{classical-quantum} source in the entanglement-assisted model, 
\aw{any rate pairs satisfying} the following inequalities are achievable: \aw{with $\alpha=\frac{2I(X:B)}{I(X:B)+I_0}$},
\begin{equation}\begin{split}
  \label{eq:funny-region}
  R_X             &\geq S(X|B),                                    \\
  R_B             &\geq \frac{1}{2}\left( S(B)+S(B|X)-I_0 \right), \\ 
  R_X +\alpha R_B &\geq S(X|B)+ \alpha S(B).
\end{split}\end{equation}

\aw{More generally,} for any auxiliary random variable $Y$ such that \aw{$Y$--$X$--$B$}
is a Markov chain, \aw{all the following rate pairs (and hence also their upper-right convex closure)} 
are achievable:
\begin{align*}
  R_X &= I(X:Y) + S(X|BY)                    = S(X|B) + I(Y:B), \\
  R_B &= \frac{1}{2}(S(B) + S(B|Y) - I(Y:W)) = S(B)-\frac{1}{2}\left(I(Y:B)+I(Y:W)\right), \\
\end{align*}
where $C$ and $W$ are the system and environment of an \aw{isometry $V:{B\rightarrow CW}$ 
with $I(W:R|Y)=0$}. 
\end{theorem}

\begin{proof}
\aw{The region described by Eq.~(\ref{eq:funny-region}) is precisely the upper-right convex
closure of the two corner points $(S(X|B),S(B))$ and $(S(X),\frac{1}{2}(S(B)+S(B|X)-I_0))$. Their
achievability} follows from \aw{Theorems} \ref{theorem: generic full rate region} and 
\ref{QSR_achievability}. 

We use the two achievable points   
$\left(S(X|B),S(B)\right)$ and $\left(S(X),\frac{1}{2}(S(B)+S(B|X)-I_0)\right)$ to show the 
\aw{second} statement. 
Namely, Alice and \aw{Debbie} (the receiver) use the Reverse Shannon Theorem to simulate the channel
taking $X$ to $Y$ in i.i.d.~fashion, which costs $I(X:Y)$ bits of classical communication \cite{Bennett2002}. 
%(XXXX  and some shared randomness, that we can remove later)
Now we are in a situation that we know, Bob has to encode $B^n$ with side information $Y^n$ at the decoder, 
which can be done \aw{at rate $\frac{1}{2}(S(B)+S(B|Y)-I(Y:W))$, by the quantum state redistribution
protocol of Theorem \ref{QSR_achievability}}. 
Then Alice has to send some more information to allow the receiver to decode $X^n$ which is an instance of 
classical compression of $X$ with quantum side information $BY$ that is already at the decoder, 
hence costing another $S(X|BY)$ in communication, \aw{by the Devetak-Winter protocol \cite{Devetak2003,Winter1999}}. 
For $Y=X$, \aw{we recover the rate point $\left(S(X),\frac{1}{2}(S(B)+S(B|X)-I_0)\right)$, 
and for $Y=\emptyset$ we recover $\left(S(X|B),S(B)\right)$.}
\end{proof}

\medskip
In Fig.~\ref{fig:full}, we show the situation for a general source, depicting
the most important inner and outer bounds on the rate region in the entanglement-assisted
model.

\begin{figure}[ht]
%  \framebox[0.8\textwidth]{\rule{0pt}{160pt}}
  \includegraphics[width=0.9\textwidth]{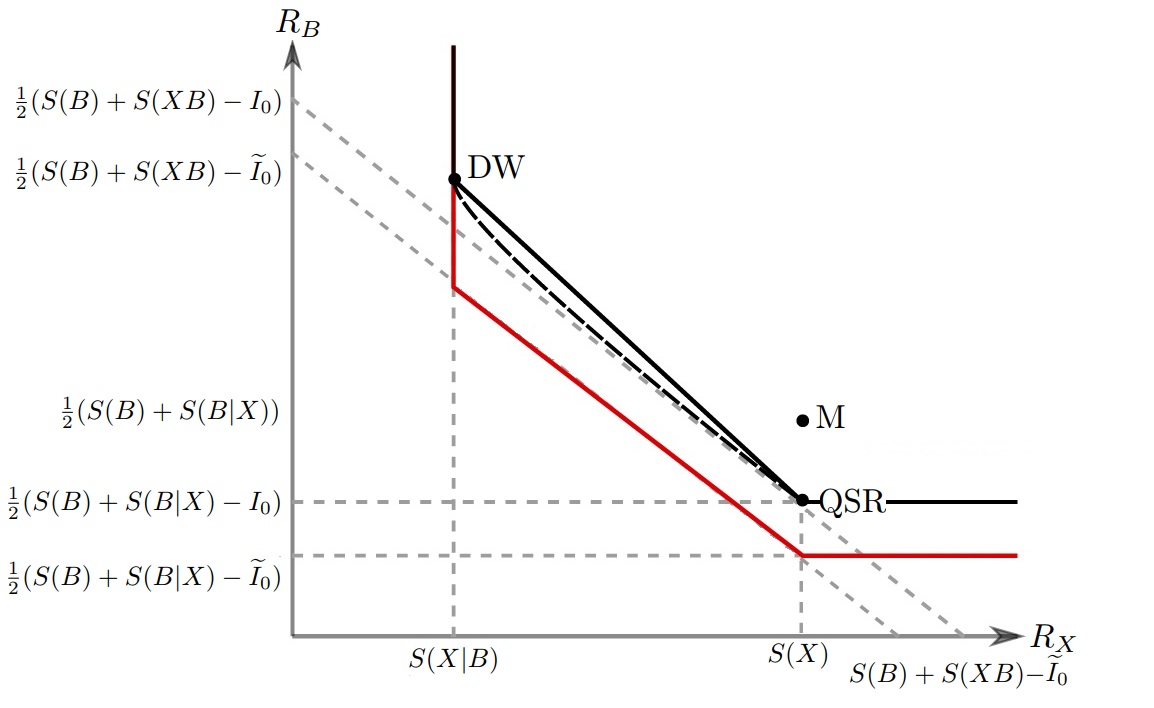}
  \caption{\aw{General outer (converse) bound, in red, and inner (achievable) bounds, in black, 
           on the entanglement-assisted rate region, assuming unlimited entanglement. 
           In general, our achievable points, the one from Devetak-Winter (DW), and the ones
           using merging (M) and quantum state redistribution (QSR) are no longer on the
           boundary of the outer bound. The achievable region is potentially slightly larger than
           the upper-right convex closure of the points DW and QSR, connected by a solid black
           straight line; indeed, the second part of Theorem \ref{thm:achieve} allows us to
           interpolate between DW and QSR along the black dashed curve}.}
  \label{fig:full}
\end{figure}

%%%%%%%%%%%%%%%%%%%%%%%%%%%%%%%%%%%%%%%%%%%%%%%%%%%%%%%%%%%%%%%%%%%%%%%%%%%%%%%%%%%%%%%%%%%%%%

\section{Discussion and Open Problems} 
\label{sec:discuss}
\aw{After seeing no progress for over 15 years in the problem of distributed
compression of quantum sources, we have decided to take a fresh look at the
classical-quantum sources considered in \cite{Devetak2003,Winter1999}. There,
the problem of compressing the classical source using the quantum part as
side information at the decoder was solved; here we were analyzing the full
rate region, in particular we were interested in the other extreme of compressing 
the quantum source using the classical part as side information at the decoder. 
Like in the classical Slepian-Wolf coding, the former problem exhibits no rate loss, 
in that the quantum part of the source is compressed to the Schumacher rate, 
the local entropy, and the sum rate equals the joint entropy of the source.
Interestingly, this is not the case for the latter problem: clearly, if the
classical side information were available both at the encoder and the decoder,
the optimal compression rate would be the conditional entropy $S(B|X)$, which
would again imply no sum rate loss. However, since the classical side information
is supposed to be present only at the decoder, we have shown that in general 
the rate sum is strictly larger, in fact generically by $\frac12 I(X:B)$, and
with this additional rate there is always a coding scheme achieving asymptotically
high fidelity. We term this additional rate ``the price of ignorance'', as it
corresponds to the absence of the side information at the encoder.

For the general case, we introduced information quantities $I_0$ and $\widetilde{I}_0$
(Definition \ref{I_delta}), to upper and lower bound the optimal quantum compression rate as
\[
  \frac12\left(S(B)+S(B|X)-\widetilde{I}_0\right) \leq R_B^* \leq \frac12\left(S(B)+S(B|X)-{I}_0\right),
\]
when unlimited entanglement is available.
For generic sources, $I_0 = \widetilde{I}_0 = 0$, but in general we do not
understand these quantities very well, and the first complex of open problems
is about them: is $I_0 = \widetilde{I}_0$ in general, or are there examples of
gaps? How to calculate either one of these quantities, given that
a priori the auxiliary register $W$ is unbounded? In fact, can one
without loss of generality put a finite bound on the dimension of $W$,
for either optimization problem?

The second open problem is about the need for prior shared entanglement to achieve
the optimal quantum compression rate $R_B^*$. As a matter of fact, it would already
be interesting to know whether the rate $\frac12\left(S(B)+S(B|X)-{I}_0\right)$ 
requires in general pre-shared entanglement.

Finally, the full rate region inherits these features: While it is simple, and
in fact generated by the optimal codes for the two compression-with-side-information 
problems (quantum compression with classical side information, and classical
compression with quantum side information), in the generic case, in general the
picture is very complicated, and we have only been able to give several outer and
inner bounds on the rate region, whose determination remains an open problem.}

%Include various observations and conjectures. \\
%Another system model with ensemble definition.\\
%Open problem: what if the ensemble of Bob's states is known.\\
%Open problem: the problem with a mixed state source $\psi_x^{BR}$ using fidelity or trace distance as error criterion.\\
%Open problem: Schumacher compression with entanglement assistance. \\
%Is entanglement is necessary to achieve the rate in assisted model?

\acknowledgments
We thank Morteza Noshad, Janis N\"otzel, Dong Yang and Odette de Crecy 
for invaluable discussions during the early stages of this project, and for
comments on the draft.

The authors acknowledge financial support from the Spanish MINECO (FIS2016-80681-P, 
FISICATEAMO no. FIS2016-79508-P, SEVERO OCHOA no.~SEV-2015-0522, FPI), 
the European Social Fund, the Fundaci\'o Cellex, the Generalitat de Catalunya 
(AGAUR grants no.~2017-SGR-1127, 2017-SGR-1341 and CERCA/Program), 
ERC AdG OSYRIS, ERC AdG IRQUAT, EU FETPRO QUIC, 
and the National Science Centre, Poland-Symfonia grant no.~2016/20/W/ST4/00314. 
Every Euro from each one of these generous institutions, 
without which we would be nothing, 
has been essential for the success of this project.

%%%%%%%%%%%%%%%%%%%%%%%%%%%%%%%%%%%%%%%%%%%%%%%%%%%%%%%%%%%%%%%%%%%%%%%%%%%%%%%%%%%%%%%%%%%%%%
%%%%%%%%%%%%%%%%%%%%%%%%%%%%%%%%%%%%%%%%%%%%%%%%%%%%%%%%%%%%%%%%%%%%%%%%%%%%%%%%%%%%%%%%%%%%%%

\appendix

\section{Miscellaneous definitions and facts}
\label{Miscellaneous_Facts}

For \aw{an} operator $X$, the \emph{trace norm}, the 
\emph{Hilbert-Schmidt norm} and 
the \emph{operator norm} are defined respectively 
in terms of $|X| = \sqrt{X^\dagger X}$:
\begin{align*}
 \|X\|_1        &= \Tr |X|, \\
 \|X\|_2        &= \sqrt{\Tr |X|^2}, \\
 \|X\|_{\infty} &= \lambda_{\max}(|X|),
\end{align*}
where $\lambda_{\max}(X)$ is the largest eigenvalue of $X$.

\begin{lemma}[Cf.~\cite{bhatia97}]
For any operator $X$,
\begin{align}%\label{norm_relations} 
\|X\|_1 \leq \sqrt{d} \|X\|_2 \leq  d \|X\|_{\infty},  
\end{align}
where $d$ equals the rank of $X$.
\qed
\end{lemma}

\begin{lemma}[Cf.~\cite{bhatia97}]
\label{norm1_trace} 
For any self-adjoint operator $X$,
\begin{align*}
  \phantom{======================:}
  \norm{X}_1=\max_{-\1 \leq Q \leq \1}\Tr(QX).
  \phantom{======================}\blacksquare
\end{align*}
\end{lemma}

\begin{lemma}[Cf.~\cite{bhatia97}]
\label{T_norm1_inequality}
For \aw{any self-adjoint operator $X$ and any operator $T$,}
\begin{align*}
  \phantom{=======================}
  \norm{TXT^{\dagger}}_1 \leq \norm{T}_{\infty}^2 \norm{X}_1.
  \phantom{======================:}\blacksquare
\end{align*}
\end{lemma}

The fidelity of two states is defined as
\begin{align*}
F(\rho,\sigma)= \Tr \sqrt{\sigma^{\frac{1}{2}} \rho \sigma^{\frac{1}{2}}}.
\end{align*}
When one of the arguments is pure, then 
\begin{align*}
  F(\rho,\ketbra{\psi}{\psi})
     =\sqrt{\Tr (\rho \ketbra{\psi}{\psi})}
     =\sqrt{\bra{\psi}\rho\ket{\psi}}.
\end{align*}

\begin{lemma}
\label{lemma:FvdG}
The fidelity is related to the trace norm as follows \cite{Fuchs1999}:
\begin{align*}
  1- F(\rho,\sigma) \leq \frac{1}{2}\|\rho-\sigma\|_1 \leq \sqrt{1-F(\rho,\sigma)^2} =: P(\rho,\sigma),
\end{align*}
where $P(\rho,\sigma)$ is the so-called purified distance,
or Battacharya distance, between quantum states.
\qed
\end{lemma}

\begin{lemma}[{Pinsker's inequality, cf.~\cite{Schumacher2002}}]
\label{lemma:Pinsker}
The trace norm and relative entropy are related by 
\begin{align*}
  \phantom{======================:}
  \|\rho-\sigma\|_1 \leq \sqrt{2 \ln 2 S(\rho\|\sigma)}. 
  \phantom{======================}\blacksquare
\end{align*}
\end{lemma}

\begin{lemma}[{Uhlmann~\cite{UHLMANN1976}}]
\label{lemma:was-dem-einen-sin-Uhlmann}
Let $\rho^A$ and $\sigma^A$ be two quantum states with fidelity $F(\rho^A,\sigma^A)$. Let $\rho^{AB}$ and $\sigma^{AC}$ be purifications of these two states, then there exists an isometry $V:{B\to C} $ such that  
\begin{align*}
  \phantom{==============}
  F\left( (\1_A \otimes V^{B\to C})\rho^{AB}(\1_A \otimes V^{B\to C})^{\dagger},\sigma^{AC} \right) = F(\rho^A,\sigma^A).
  \phantom{=============:}\blacksquare
\end{align*}
\end{lemma}

A consequence of this, due to \cite[Lemma~2.2]{Devetak2008_1}, is as follows.
\begin{lemma}
Let $\rho^A$ and $\sigma^A$ be two quantum states with trace distance 
$\frac12 \|\rho^A-\sigma^A\|_1 \leq \epsilon$, and
let $\rho^{AB}$ and $\sigma^{AC}$ be purifications of these two states.
Then there exists an isometry $V:{B\to C}$ such that  
\begin{align*}
  \phantom{=============:}
  \left\| (\1_A \otimes V^{B\to C})\rho^{AB} (\1_A \otimes V^{B\to C})^{\dagger}
                                    - \sigma^{AC} \right\|_1  \leq \sqrt{\epsilon(2-\epsilon)} \,.
  \phantom{=============}\blacksquare
\end{align*}
\end{lemma}

\begin{lemma}[{Fannes~\cite{Fannes1973}; Audenaert~\cite{Audenaert2007}}]
\label{Fannes-Audenaert inequality}
Let $\rho$ and $\sigma$ be two states on $d$-dimensional space with trace distance 
$\frac12\|\rho-\sigma\|_1 \leq \epsilon$, then
\begin{align*}
  |S(\rho)-S(\sigma)| \leq \epsilon\log d + h(\epsilon),
\end{align*}
where $h(\epsilon)=-\epsilon \log \epsilon -(1-\epsilon)\log (1-\epsilon)$ \aw{is the binary entropy}.
\end{lemma}

\aw{There is also an extension} of the Fannes inequality for the conditional entropy; 
this lemma is very useful \aw{especially} when the dimension of the 
system \aw{conditioned on} is unbounded.

\begin{lemma}[{Alicki-Fannes~\cite{Alicki2004}; Winter~\cite{Winter2016}}]
\label{AFW lemma}
Let $\rho$ and $\sigma$ be two states on a bipartite Hilbert space 
$A\otimes B$ with trace distance $\frac12\|\rho-\sigma\|_1 \leq \epsilon$, then
\begin{align*}
  \phantom{==================}
  |S(A|B)_{\rho}-S(A|B)_{\sigma}| \leq 2\epsilon \log |A| + 2h(\epsilon).
  \phantom{==================}\blacksquare
\end{align*}
\end{lemma}

\begin{lemma}
\label{full_support_lemma}
Let $\rho$ be a state with full support on the Hilbert space \aw{$A$}, i.e.~it 
\aw{has positive minimum} eigenvalue $\lambda_{\min}$, and
let $\ket{\psi}^{AR}$ be a purification of $\rho$ on the Hilbert space \aw{$A \otimes R$}. 
Then any purification of another state $\sigma$ on \aw{$A$} is of the form 
\begin{align*}
  (\1_A \otimes T) \ket{\psi}^{AR},
\end{align*}
where $T$ is an operator acting on system $R$ with $\| T \|_{\infty} \le \frac{1}{\sqrt{\lambda_{\min}}}$.
\begin{proof}
Let $\rho=\sum_i \lambda_i \ketbra{e_i}{e_i}$ and $\sigma=\sum_j \mu_j \ketbra{f_j}{f_j}$ be spectral decompositions  of the states. The purification of $\rho$  is  $\ket{\psi}^{AR}=\sum_i \sqrt{\lambda_i} \ket{e_i} \ket{i}$. Define $\ket{\phi}^{AR}=\sum_j \sqrt{\mu_j} \ket{f_j} \ket{j}$.
Any purification of the state $\sigma$ is of the form  
$\1_A \otimes V \ket{\phi}^{AR}$ where $V$ is an isometry acting on system $R$. 
Write the eigenbasis $\set{\ket{f_j}}$ as linear combination of eigenbasis $\set{\ket{e_j}}$. Then, we have $\ket{\phi}^{AR}=\sum_{i,j} \sqrt{\mu_j} \alpha_{ij} \ket{e_i} \ket{j}$. Define the operator  $P=\sum_{jk} p_{jk} \ketbra{j}{k}$ where $p_{jk}=\alpha_{kj} \sqrt{\frac{\mu_j}{\lambda_k}}$. It is immediate to see that
\begin{align*}
    \ket{\phi}^{AR}=(\1_A \otimes P) \ket{\psi}^{AR}. 
\end{align*}
Thus, we have $(\1_A \otimes V) \ket{\phi}^{AR} = (\1_A \otimes VP) \ket{\psi}^{AR}$. 
Defining $T=VP$, we then have 
\begin{align*}
\lambda_{\max} (T^{\dagger}T)=\lambda_{\max} (P^{\dagger}P)   \leq \Tr (P^{\dagger}P) 
 =\sum_{j,k}|p_{jk}|^2  
 = \sum_{j,k}\frac{|\alpha_{kj}|^2\mu_j}{\lambda_k}  
 \leq \frac{1}{\lambda_{\min}},
\end{align*}
where the last inequality follows \aw{from the orthonormality of} the basis $\set{\ket{f_j}}$. 
\end{proof}
\end{lemma}

\section{Proof of Lemma \ref{decoupling condition} (decoupling condition)}
\label{decoupling_condition_proof}
In this subsection, we show that the fidelity criterion (\ref{F_QCSW_assisted}) 
implies that given $x^n$, the environments $W_X$, $W_B$ and $W_D$ of \aw{Alice's, Bob's and \aw{Debbie}'s isometries
are decoupled from the the rest of the output systems. For convenience, we restate the lemma
we are aiming to prove.}
%decoded systems and the reference $\hat{X}^n \hat{B}^nR^n$. 
%that we call decoupling condition.

\medskip\noindent
\textbf{Lemma \ref{decoupling condition}.}\ \emph{(Decoupling condition)}\
\aw{For a code of block length $n$ and error $\epsilon$ in the entanglement-assisted model,  
let $W_X$, $W_B$ and $W_D$ be the environments of Alice's and Bob's encoding and of
Debbie's decoding isometries, respectively. Then,
\[
  I(W_XW_BW_D B_0'D_0':\hat{X}^n\hat{B}^nR^n|{X'}^n)_\xi \leq n \delta(n,\epsilon) ,
\] 
where $\delta(n,\epsilon) = 4\sqrt{6\epsilon} \log(|X| |B|) + \frac2n h(\sqrt{6\epsilon})$, 
with the binary entropy $h(\epsilon)=-\epsilon\log\epsilon - (1-\epsilon)\log(1-\epsilon)$;
the mutual information is with respect to the state 
\[
  \xi^{{X'}^n \hat{X}^n \hat{B}^n B_0' D_0' W_XW_BW_D R^{n}}
      =\left(\mathcal{D} \circ (\mathcal{E}_X \otimes \mathcal{E}_{B} \otimes \id_{D_0}) \otimes \id_{{X'}^n R^n}\right) 
        \omega^{X^n {X'}^n B^n R^n } \otimes \Phi_K^{B_0D_0}.
\]}

\begin{proof}
The parties share $n$ copies of the state $\omega^{X^{\prime} X B R}$, where 
Alice and Bob have access to systems $X^n$ and $B^n$, respectively, and ${X'}^n$ and $R^n$ are the reference systems.  
Alice and Bob apply the following isometries to encode their systems, respectively:
\begin{align*}
    U_{X}:{X^n}         &\longrightarrow {C_X W_X},  \\
    U_{B}:{B^n B_0}     &\longrightarrow {C_B B_0' W_B},                   
%    V_{D}:{C_X C_B D_0} &\longrightarrow {\hat{X}^n \hat{B}^n W_D D_0'},
\end{align*}
where Alice and Bob send respectively their compressed information $C_X$ and $C_B$ to \aw{Debbie} 
and keep the environment parts $W_X$ and $W_B$ of their respective isometries for themselves. 
%The following state shows the encoded state
%\begin{align*}
%  \sigma^{X'^n C_X W_X C_B W_B B_0' D_0 R^n}
%  =\sum_{x^n} p(x^n)\ketbra {x^n}^{X'^n} \otimes  \ketbra{\nu_{x^n}}^{ C_X W_X} \otimes \ketbra{\phi_{x^n}}^{ C_B W_B B_0' D_0 R^n}.    
%\end{align*}
\aw{Debbie} applies \aw{the} decoding isometry \aw{$V:{C_X C_B D_0} \longrightarrow {\hat{X}^n \hat{B}^n D_0' W_D}$ 
to the systems $C_XC_B$ and her part of the entanglement $D_0$,
to generate the output systems $\hat{X}^n \hat{B}^n D_0'$, with $W_D$ the environment of her isometry.
This leads to the following final state after decoding:
\begin{align*}
  \xi^{X'^n \hat{X}^n \hat{B}^n B_0'D_0' W_X W_B W_D R^n}
      =\sum_{x^n} p(x^n) \ketbra{x^n}^{X'^n} \otimes \ketbra{\xi_{x^n}}^{\hat{X}^n \hat{B}^n B_0'D_0' W_X W_B W_D R^n},     
\end{align*}
where 
\[
  \ket{\xi_{x^n}}^{\hat{X}^n \hat{B}^n B_0'D_0' W_X W_B W_D R^n} 
     \!\!= V^{{C_XC_BD_0 \to \hat{X}^n\hat{B}^nD_0'W_D}}
        \big( U_X^{X^n \to C_XW_X}\!\!\ket{x^n}^{X^n} 
              \!\otimes U_B^{B^nB_0 \to C_BB_0'W_B} (\ket{\psi_{x^n}}^{B^nR^n}\!\ket{\Phi_K}^{B_0D_0}) \bigr).
\]}

The fidelity defined in \aw{Eq.}~(\ref{F_QCSW_assisted}) is \aw{now} bounded as follows:
\begin{align} 
  \label{eq-B1}
  \overline{F} 
    &= F\left(\omega^{X^n X'^n B^n R^n} \otimes \Phi_L^{B_0'D_0'},
               \left(\mathcal{D} \circ (\id_{X^nD_0} \otimes \mathcal{E}_{B}) \otimes \id_{X'^n R^n} \right) 
                                                   \omega^{X^n X'^n B^n R^n } \otimes \Phi_K^{B_0D_0} \right) \nonumber \\
    &= F\left(\omega^{X'^n X^n B^n R^n} \otimes \Phi_L^{B_0'D_0'},
               \xi^{X'^n \hat{X}^n \hat{B}^n B_0' D_0' R^n} \right)                                  \nonumber \\  
   &\leq F\left(\omega^{X'^n X^n B^n R^n}, \xi^{X'^n \hat{X}^n \hat{B}^n R^n} \right)                          \nonumber \\  
   &=    \sum_{x^n \in \mathcal{X}^n} p(x^n) F\left(\ketbra{x^n}{x^n}^{X^n} \otimes \ketbra{\psi_{x^n}}{\psi_{x^n}}^{B^nR^n},
                                                     \xi_{x^n}^{\hat{X}^n \hat{B}^n R^n} \right)               \nonumber \\
   &=    \sum_{x^n} p(x^n) \sqrt{\bra{x^n}\bra{\psi_{x^n}}^{B^n R^n} 
                                  \xi_{x^n}^{\hat{X}^n\hat{B}^nR^n} \ket{x^n}\ket{\psi_{x^n}}^{B^n R^n}}       \nonumber \\
   &\leq \sum_{x^n} p(x^n) \sqrt{\| \xi_{x^n}^{\hat{X}^n\hat{B}^n R^n} \|},
\end{align}
where the \aw{inequality in the third line} is due to the monotonicity of fidelity under partial trace, 
and \aw{$\|\xi_{x^n}^{\hat{X}^n \hat{B}^n R^n}\|$ denotes the operator norm, which in this case
of a positive semidefinite operator is the maximum eigenvalue of $\xi_{x^n}^{\hat{X}^n \hat{B}^n R^n}$. 
Now, consider the Schmidt decomposition of the state 
$\ket{\xi_{x^n}}^{\hat{X}^n \hat{B}^n B_0'D_0' W_X W_B W_DR^n}$ with respect to the partition
$\hat{X}^n \hat{B}^n R^n$ : $B_0' D_0'W_X W_B W_D$, i.e.
%state $\ket{\xi_{x^n}}^{\hat{X}^n \hat{B}^n B_1 D_1 W_B W_DR^n}$ on which is as follows 
\begin{align*}
\ket{\xi_{x^n}}^{\hat{X}^n \hat{B}^n B_0'D_0' W_X W_B W_DR^n}
     = \sum_{i} \sqrt{\lambda_{x^n}(i)}\ket{v_{x^n}(i)}^{\hat{X}^n \hat{B}^n R^n} \ket{w_{x^n}(i)}^{B_0'D_0' W_X W_B W_D}. 
\end{align*}}
High average fidelity $\overline{F} \geq 1-\epsilon$ implies that \emph{on average} 
the above state has Schmidt rank \aw{approximately one. In other words, the two subsystems are 
nearly independent:} 
\begin{align}
  \label{eq:almost-pure}
  \sum_{x^n} p(x^n) F&\left( \ketbra{\xi_{x^n}}{\xi_{x^n}}^{\hat{X}^n \hat{B}^n B_0'D_0' W_XW_B W_DR^n},
                             \xi_{x^n}^{\hat{X}^n \hat{B}^n R^n} \otimes \xi_{x^n}^{B_0'D_0' W_X W_B W_D} \right) \nonumber\\
    &= \sum_{x^n} p(x^n) \sqrt{\bra{\xi_{x^n}}
                                {\xi_{x^n}^{\hat{X}^n \hat{B}^n R^n} \otimes \xi_{x^n}^{B_0' D_0' W_X W_B W_D} 
                               \ket{\xi_{x^n}}}}                                                                 \nonumber\\
    &= \sum_{x^n} p(x^n) \sum_i \lambda_{x^n}(i)^{\frac32}                                                 \nonumber\\
    &\geq \sum_{x^n} p(x^n) \|\xi_{x^n}^{\hat{X}^n \hat{B}^n R^n}\|^{\frac32}                        \nonumber\\
    &\geq \left( \sum_{x^n} p(x^n) \sqrt{\|\xi_{x^n}^{\hat{X}^n \hat{B}^n R^n}\|} \right)^{3}  \nonumber\\
    &\geq (1-\epsilon)^3 
     \geq 1 - 3\epsilon,  
\end{align}
where the \aw{inequality in the fifth line} follows from the convexity of $x^3$ for $x \geq 0$, 
and \aw{in the sixth line we have used Eq.}~(\ref{eq-B1}). 
Based on the relation between fidelity and trace distance \aw{(Lemma \ref{lemma:FvdG}), 
we thus obtain for the product ensemble
\[
  \zeta^{X'^n \hat{X}^n \hat{B}^n B_0'D_0' W_X W_B W_D R^n} 
    := \sum_{x^n} p(x^n)\ketbra{x^n}^{X'^n} \otimes \xi_{x^n}^{\hat{X}^n \hat{B}^n R^n} \otimes \xi_{x^n}^{B_0'D_0' W_X W_B W_D},
\] 
that
\begin{align*}
  \| \xi - \zeta \|_1 
    &=    \sum_{x^n} p(x^n) \norm{\ketbra{\xi_{x^n}}{\xi_{x^n}}^{\hat{X}^n \hat{B}^n B_0'D_0' W_XW_B W_DR^n}
                                      \!-\! \xi_{x^n}^{\hat{X}^n \hat{B}^n R^n} \!\otimes\! \xi_{x^n}^{B_0'D_0' W_X W_B W_D}}_1 \\
    &\leq 2\sqrt{6\epsilon}.   
\end{align*}}
By \aw{the Alicki-Fannes inequality (Lemma \ref{AFW lemma}), this implies
\begin{align}
  \label{decoupling_I}
  I(\hat{X}^n\hat{B}^nR^n : B_0'D_0'W_XW_B W_D | {X'}^n)_\xi 
     &=    S(\hat{X}^n \hat{B}^n R^n | {X'}^n)_\xi - S(\hat{X}^n \hat{B}^n R^n | {X'}^n B_0'D_0' W_XW_B W_D)_\xi \nonumber \\ 
     &\leq 2\sqrt{6\epsilon} \log(|X|^n |B|^n |R|^n) + 2 h(\sqrt{6\epsilon}) \nonumber \\
     &\leq 2\sqrt{6\epsilon} \log(|X|^{2n} |B|^{2n}) + 2 h(\sqrt{6\epsilon}) 
           =: n \delta(n,\epsilon),  
\end{align}
where we note in the second line that 
$S(\hat{X}^n \hat{B}^n R^n|X'^n B_0'D_0' W_XW_B W_D)_\zeta 
 = S(\hat{X}^n \hat{B}^n R^n)_\zeta = S(\hat{X}^n \hat{B}^n R^n)_\xi$,
and in the third line that we can without loss of generality assume $|R| \leq |X| |B|$, 
since that is the maximum possible dimension of the support of $\omega^R$.}
\end{proof}

%\bibliography{qcsw.bib}

\end{document}